\documentclass[onecolumn,floats,floatfix,showpacs,nofootinbib,11pt]{revtex4-1}

\usepackage{graphicx}
\usepackage{dcolumn}
\usepackage{bm}
\usepackage{graphics}
\usepackage{slashed}
\usepackage{amssymb}
\usepackage{natbib}
\usepackage{amsmath}
\usepackage{url}
\usepackage[usenames,dvipsnames,svgnames,table]{xcolor}
\usepackage{amsfonts}
\usepackage{subfig}
\usepackage{amsthm}
\usepackage{hyperref}
\usepackage[retainorgcmds]{IEEEtrantools}

\newcommand{\bea}{\begin{eqnarray}}
\newcommand{\ena}{\end{eqnarray}}
\newcommand{\bean}{\begin{eqnarray*}}
\newcommand{\enan}{\end{eqnarray*}}

\newtheorem{theorem}{Theorem}
\newtheorem{lemma}{Lemma}

\newtheorem{definition}{Definition}
\newtheorem{corollary}{Corollary}

\usepackage{color}
\newcommand*{\mathcolor}{}
\def\mathcolor#1#{\mathcoloraux{#1}}
\newcommand*{\mathcoloraux}[3]{%
  \protect\leavevmode
  \begingroup
    \color#1{#2}#3%
  \endgroup
}

\input{tcilatex}

\begin{document}

\title{A necessary condition for strong hyperbolicity of general first order systems.}
\author{Fernando Abalos${}^{1}$}
\email{jfera18@gmail.com}

\affiliation{${}^{1}$ Facultad de Matem\'atica, Astronom\'\i{}a y F\'\i{}sica, Universidad Nacional de C\'ordoba and IFEG-CONICET, Ciudad Universitaria, X5016LAE C\'ordoba, Argentina }

\begin{abstract}

We study strong hyperbolicity of first order partial differential equations
for systems with differential constraints. In these cases, the number of
equations is larger than the unknown fields, therefore, the standard Kreiss
necessary and sufficient conditions of strong hyperbolicity do not directly
apply. To deal with this problem one introduces a new tensor, called a
reduction, which selects a subset of equations with the aim of using them as
evolution equations for the unknown. If that tensor leads to a strongly
hyperbolic system we call it a hyperbolizer. There might exist many of them
or none.

A question arises on whether a given system admits any hyperbolization at
all. To sort-out this issue, we look for a condition on the system, such
that, if it is satisfied, there is no hyperbolic reduction. To that purpose
we look at the singular value decomposition of the whole system and study
certain one parameter families ($\varepsilon $) of perturbations of the
principal symbol. We look for the perturbed singular values around the
vanishing ones and show that if they behave as $O\left( \varepsilon
^{l}\right) $, with $l\geq 2$, then there does not exist any hyperbolizer.
In addition, we further notice that the validity or failure of this
condition can be established in a simple and invariant way.

Finally we apply the theory to examples in physics, such as Force-Free
Electrodynamics in Euler potentials form and charged fluids with finite
conductivity. We find that they do not admit any hyperbolization.

\end{abstract}


\maketitle


 

\section{INTRODUCTION.\label{cap2}}

Following \cite{geroch1996partial}, we consider a first order system of
partial differential equations on a fiber bundle $b$ (real or complex) with
base manifold $M$ (real) of dimension $n$ 
\begin{equation}
\mathfrak{N}_{~\alpha }^{Aa}\left( x,\phi \right) \nabla _{a}\phi ^{\alpha
}=J^{A}\left( x,\phi \right)  \label{sht_1}
\end{equation}

Here $M$ is the space-time and $x$ are points of it. We call $X_{x}$ the
fiber of $b$ at point $x$ and its dimension $u$. A cross section $\phi $ is
a map from open sets of $M$ to $b$, i.e. $\phi :U\rightarrow b$, they are
the unknown fields. Here $\mathfrak{N}_{~\alpha }^{Aa}$ and $J^{A}$ are
giving fields on $b,$ called the principal symbol and the current of the
theory respectively. That is, they do not depend of the derivative of $\phi $%
, but can depend on $\phi $ and $x.$ The multi-tensorial index $A$ belongs
to a new vector space $E_{x}$ that indicates the space of equations. We call
the dimension of this space $e$, and from now on we shall assume it is equal
or greater than the dimension of $X_{x}$ i.e.,, $e=\dim "A"\geq \dim "\alpha
"=u.$

In many examples of physical interest, system (\ref{sht_1}) can be splitted
into evolution and constraint equations. The first ones define an initial
value problem, namely, they are a set of equations, such that, given data $%
\phi _{0}^{\alpha }=\left. \phi ^{\alpha }\right\vert _{S}$ over a specific
hypersurface $S$ of dimension $n-1$, they determine a unique solution in a
neighborhood of $S$. The second ones restrict the initial data and have to
be fulfilled during evolution. For a detailed discussion see Reula's work 
\cite{reula2004strongly}.

The choice of a coherent set of evolution equations is made in terms of a
new map, $h_{~A}^{\alpha }\left( x,\phi \right) :E_{x}\rightarrow X_{x}$
called a \textbf{reduction}. It takes a linear combination of the whole set
of equations (\ref{sht_1}) and reduces them to a set of dimension $u$, which
will be used for the initial value problem, 
\begin{equation}
h_{~A}^{\alpha }\mathfrak{N}_{~\gamma }^{Aa}\nabla _{a}\phi ^{\gamma
}=h_{~A}^{\alpha }J^{A}\left( x,\phi \right)   \label{eq_evol}
\end{equation}

We would like the above system to be well posed and stable under arbitrary
choices of sources (see \cite{geroch1996partial}, \cite{reula1998hyperbolic}%
, \cite{gustafsson1995time}, \cite{reula2004strongly}, \cite%
{kreiss2004initial}, \cite{hilditch2013introduction}, \cite%
{hadamard2014lectures}, \cite{kreiss2002some}), for that, we shall need for
the tensors $\mathfrak{N}_{~\gamma }^{Aa}$\ and $h_{~A}^{\alpha }\mathfrak{N}%
_{~\gamma }^{Aa}$ to satisfy certain properties which we display in the
following definitions.

Given $\omega _{a}\in T_{x}M^{\ast }$ consider the set of planes $S_{\omega
_{a}}^{%
\mathbb{C}
}=$\{$n_{a}\left( \lambda \right) :=-\lambda \omega _{a}+\beta _{a}$\ for $%
\lambda \in 
\mathbb{C}
$\ and all other covector $\beta _{a}\in T_{x}M^{\ast }$\ not proportional
to $\omega _{a}$\}\footnote{%
This set turns into a set of lines when $\lambda $ run over $%
\mathbb{R}
$ and we call it $S_{\omega _{a}}$.}.\textbf{\ }So, following the covariant
formulation of Reula \cite{reula2004strongly} and Geroch \cite%
{geroch1996partial}, we need to study the kernel of $\mathfrak{N}_{~\gamma
}^{Aa}n\left( \lambda \right) _{a}$ with $n\left( \lambda \right) _{a}\in
S_{\omega _{a}}^{%
\mathbb{C}
}$.

\begin{definition}
System (\ref{sht_1}) is \textbf{hyperbolic} at the point $\left( x,\phi
\right) $, if there exists $\omega _{a}\in T_{x}M^{\ast }$ such that for
each plane $n\left( \lambda \right) _{a}$ in $S_{\omega _{a}}^{%
\mathbb{C}
}$, the principal symbol $\mathfrak{N}_{~\gamma }^{Aa}n\left( \lambda
\right) _{a}$ can only have a non trivial kernel when $\lambda $ is real.
\end{definition}

An important concept for hyperbolic systems are their \textbf{characteristic
surfaces}, they are the set of all covectors $n_{a}\in T_{x}M^{\ast }$ such
that $\mathfrak{N}_{~\gamma }^{Aa}n_{a}$ has non trivial kernel.

The hyperbolicity condition is not sufficient for well posedness, we
strengthen it.

\begin{definition}
\label{def_strongly}$\footnote{%
That the well posedness property follows from studying the hyperbolicity can
be seen by considering a high frequency limit perturbation of a background
solution of (\ref{eq_evol}) as $\tilde{\phi}^{\alpha }=\phi ^{\alpha
}+\varepsilon \delta \phi ^{\alpha }e^{i\frac{f\left( x\right) }{\varepsilon 
}}$ with $\varepsilon $ approaching zero, and resulting in a equation for $%
\delta \phi ^{\alpha }$ and $n_{a}:=\nabla _{a}f$ (see \cite%
{friedrichs1954symmetric})}$ System (\ref{sht_1}) is \textbf{strongly
hyperbolic} at $\left( x,\phi \right) $ (some background solution) if there
exist a covector $\omega _{a}$ and a reduction $h_{~A}^{\alpha }\left(
x,\phi \right) $, such that:

$i)$ $A_{~\gamma }^{\alpha a}\omega _{a}:=h_{~A}^{\alpha }\mathfrak{N}%
_{~\gamma }^{Aa}\omega _{a}$ is invertible, and

$ii)$ For each $n\left( \lambda \right) _{a}\in S_{\omega _{a}}$,%
\begin{equation}
\dim \left( span\left\{ \underset{\lambda \in 
\mathbb{R}
}{\cup }Ker\left\{ A_{~\gamma }^{\alpha a}n\left( \lambda \right)
_{a}\right\} \right\} \right) =u  \label{ker_1}
\end{equation}
\end{definition}

When this happens we refer to reduction $h_{~A}^{\gamma }$ as a \textbf{%
hyperbolizer}. Notice that when the system is strongly hyperbolic, it is
hyperbolic.

Note that from $i,$ $h_{~A}^{\alpha }$ is surjective and there exists $%
\omega _{a}$ such that $\mathfrak{N}_{~\gamma }^{Aa}\omega _{a}$ has no
kernel. And because $A_{~\gamma }^{\alpha a}\omega _{a}$ is invertible, the
set $\left\{ \lambda _{i}\right\} $, such that $A_{~\gamma }^{\alpha
a}n\left( \lambda _{i}\right) _{a}=A_{~\gamma }^{\alpha a}\beta _{a}-\lambda
_{i}A_{~\gamma }^{\alpha a}\omega _{a}$ has kernel, are the eigenvalues of $%
\left( A_{~\beta }^{\alpha a}\omega _{a}\right) ^{-1}A_{~\gamma }^{\beta
a}\beta _{a}$ and they are functions of $\omega _{a}$ and $\beta _{a}.$
Condition $ii$ request that these must be real and $\left( A_{~\beta
}^{\alpha a}\omega _{a}\right) ^{-1}A_{~\gamma }^{\beta a}\beta _{a}$
diagonalizable for any $\beta _{a}$. These eigenvalues are given by the
roots of the polynomial equation $\det \left( h_{~A}^{\alpha }\mathfrak{N}%
_{~\beta }^{Aa}n\left( \lambda \right) _{a}\right) =0$ and the solution $%
n\left( \lambda \right) _{a}$ are called \textbf{characteristic surfaces of the
evolution equations}.

Therefore, an important question is: What are necessary and sufficient
conditions for the principal symbol $\mathfrak{N}_{~\gamma
}^{Aa}n_{a}:X_{x}\rightarrow E_{x}$ to admit a hyperbolizer?

We find a partial answer, namely an algebraic necessary condition (and
sufficient condition for the case without constraints), which is of
practical importance for ruling out theories as unphysical, when they do not
satisfy it.

In general the hyperbolizer depends also on $\beta _{a}$, in that case,
equation (\ref{eq_evol}) becomes a pseudo-differential expression and it is
necessary for it to be well posedness that the hyperbolizer is smooth also
in $\beta _{a}$ \footnote{%
Unlike the usual terminology, we exclude the smoothness condition from
definition \ref{def_strongly}.}. However, smoothness property shall not play
any role in what follows, but it is an issue that should be addressed at
some point of the development of the theory.

To find this condition we shall use the Singular Value Decomposition SVD (we
give a covariant formalism of SVD in appendix \ref{appendix}) of $\mathfrak{N%
}_{~\gamma }^{Aa}n_{a}$ in neighborhood of the characteristic surfaces, and
conclude that the way in which the singular values approach this surfaces
gives information about the size of the kernel.

\section{MAIN RESULTS\label{Main_results}.}

In this section we introduce our main results. Consider any fixed $\theta
\in \left[ 0,2\pi \right] $ and a line $n\left( \lambda \right) _{a}\in
S_{\omega _{a}}$ for some $\omega _{a}$. So we define the extended
two-parameter line $n_{\varepsilon ,\theta }\left( \lambda \right)
_{a}=-\varepsilon e^{i\theta }\omega _{a}+n_{a}\left( \lambda \right) $ with 
$\varepsilon $ real and $0\leq \left\vert \varepsilon \right\vert <<1.$ Then
the perturbed principal symbol results in 
\begin{equation}
\mathfrak{N}_{~\beta }^{Aa}n_{\varepsilon ,\theta }\left( \lambda \right)
_{a}=\left( -\lambda \mathfrak{N}_{~\beta }^{Aa}\omega _{a}+\mathfrak{N}%
_{~\beta }^{Aa}\beta _{a}\right) -\varepsilon \left( e^{i\theta }\mathfrak{N}%
_{~\beta }^{Aa}\omega _{a}\right) .  \label{eq_per_1}
\end{equation}

Moro et. al. \cite{moro2002first} and Soderstrom \cite%
{soderstrom1999perturbation} proved that the singular values of this
perturbed operator admit a Taylor expansion at least up to second order in $%
\left\vert \varepsilon \right\vert ,$\ also they showed explicit formulas to
calculate them (see theorem \ref{TTor} in appendix \ref{appendix_2}). We use
their results to prove a necessary condition for strong hyperbolicity.

Consider first the case that no constraints are present. This is $\dim
"A"=\dim "\alpha "$, and all equations should be considered as evolution
equations. We call it, the "Square" case, since the principal symbol $%
\mathfrak{N}_{~\gamma }^{Aa}n_{a}$, maps between spaces of equal dimensions,
and hence it is a square matrix.

In this case, any invertible reduction tensor $h_{~A}^{\alpha }$ that we
use, would keep the same kernels. Thus strong hyperbolicity is a sole
property of the principal symbol.

For this type of systems the Kreiss's Matrix Theorem (see theorem \ref%
{Kreiss_theorem}) \cite{gustafsson1995time}\cite{kreiss2004initial}\cite%
{kreiss1962stabilitatsdefinition} lists several necessary and sufficient
conditions for strong hyperbolicity. In subsection \ref{square_operators_1}
we shall prove the theorem below, which incorporates to the Kreiss's Matrix
Theorem a further necessary and sufficient condition.

\begin{theorem}
\label{Teo_no_jordan}System (\ref{sht_1}) with $\dim "A"=\dim "\alpha "$ is
strongly hyperbolic if and only if the following conditions are valid:

1- There exists $\omega _{a}$ such that the system is hyperbolic and $%
\mathfrak{N}_{~\gamma }^{Aa}\omega _{a}$ has no kernel.

2- For each line $n\left( \lambda \right) _{a}$ in $S_{\omega _{a}}$
consider any extended one $n_{\varepsilon ,\theta }\left( \lambda \right)
_{a}$ then the principal symbol $\mathfrak{N}_{~\beta }^{Aa}n_{\varepsilon
,\theta }\left( \lambda \right) _{a}$ has only singular values of orders $%
O\left( \left\vert \varepsilon \right\vert ^{0}\right) $ and $O\left(
\left\vert \varepsilon \right\vert ^{1}\right) $.
\end{theorem}

In general we consider systems that fulfill 1, and we refer to 2 as "the
condition for strong hyperbolicity".

We shall also give a couple of examples on how to apply these results: A
simple matrix case of $2\times 2$, in subsection \ref{Example_matrix}, and a
physical example, charged fluids with finite conductivity in subsection \ref%
{Example_fluids}. We shall show that conductivity case is only weakly
hyperbolic.

Consider now $\dim "A">\dim "\alpha "$. In this case, we want to find a
suitable subset of evolution equations. In general if we consider $n\left(
\lambda \right) _{b}\in S_{\omega _{a}}$ and count the dimension of the
kernel of \ $\mathfrak{N}_{~\gamma }^{Bb}n\left( \lambda \right) _{b}$ (the
physical propagation directions), over $\lambda \in 
\mathbb{R}
$, we find that this number is less than $u$. As a consequence we need to
introduce a hyperbolizer in order to increase the kernel and fulfill
condition (\ref{ker_1}).

We call it the "rectangular case" and we find only a necessary condition for
strong hyperbolicity:

\begin{theorem}
\label{sht_Theo_1} When $\dim "A">\dim "\alpha "$ in system \ref{sht_1},
conditions in Theorem \ref{Teo_no_jordan} are still necessary.
\end{theorem}

As we said before, this condition has practical importance since it can be
checked with a simple calculation (see theorem \ref{Conjetura_final}), thus
discarding as unphysical those systems that do not satisfy it. We prove this
theorem in subsection \ref{Rectangular_operators}, and present its application to a
physically motivated example, namely Force Free electrodynamics in Euler
potentials description, in subsection \ref{Euler_potenciales}. We shall show
that this system does not admit a hyperbolizer (it is weakly hyperbolic) for
any choice of reduction, and we emphasize how simple it is to show that one
of their singular values is order $O\left( \left\vert \varepsilon
\right\vert ^{l}\right) $ with $l\geq 2$, using theorem \ref{Conjetura_final}
in section \ref{chap_examples}.


\section{SINGULAR VALUE DECOMPOSITION, PERTURBATION THEORY AND
DIAGONALIZATION OF LINEAR OPERATOR. \label{singular_value_}}

In this section we shall use the Singular Value Decomposition (SVD) to find
conditions for Jordan diagonalization. Those will be used to prove Theorems %
\ref{Teo_no_jordan} and \ref{sht_Theo_1}, obtaining conditions for strong
hyperbolicity. In appendix \ref{appendix} we describe the SVD theory in
detail. We included it because our approach to the topic is a bit different
than the standard one, as presented in the literature.

In order to prove our main results we shall study the principal symbol $%
\mathfrak{N}_{~\gamma }^{Aa}\left( x,\phi \right) n\left( z\right) _{a}=-z%
\mathfrak{N}_{~\alpha }^{Ba}\omega _{a}+\mathfrak{N}_{~\alpha }^{Ba}\beta
_{a}$ with $n\left( z\right) _{a}\in S_{\omega _{a}}^{%
\mathbb{C}
}$ \footnote{%
Notice that we changed $\lambda $ for $z$ to remember that $z$ belong to $%
\mathbb{C}
.$} for some $\omega _{a}\in T_{x}^{\ast }M$, and perturbations as in
equation (\ref{eq_per_1}). We shall assume that there exists $\omega _{a}$
such that $\mathfrak{N}_{~\alpha }^{Bb}\omega _{b}$ has no kernel and show a
necessary (and sufficient in the square case) condition for the existence of
a reduction $h_{~A}^{\alpha }$ such that $\left( A_{~\mu }^{\alpha a}\omega
_{a}\right) ^{-1}A_{~\gamma }^{\mu a}\beta _{a}$ is diagonalizable, even
with complex eigenvalues; recall that \ $A_{~\gamma }^{\alpha
a}=h_{~A}^{\alpha }\mathfrak{N}_{~\gamma }^{Aa}$. In addition, if we also
request that the system is hyperbolic with this $\omega _{a}$, we would have
completed the above theorems.

In what follows, we present the notation that we will use through the paper.
We shall name square and rectangular operators to those that map spaces of
equal or different dimensions respectively. We call $K_{~\alpha }^{A}\left( x,\phi ,\beta _{a}\right) :=\mathfrak{N}%
_{~\alpha }^{Ba}\beta _{a}$ and $B_{~\alpha }^{A}\left( x,\phi ,\omega
_{a}\right) :=\mathfrak{N}_{~\alpha }^{Bb}\omega _{b}$. Notice that these
operators change with\textbf{\ }$x,\phi ,\beta _{a}$\textbf{\ }and\textbf{\ }%
$x,\phi ,\omega _{a}$ respectively\textbf{. }However,\textbf{\ }the
condition we are looking for are algebraic, so they hold at each particular
point, which we shall assert from now on. In addition, we define%
\begin{equation}
T_{~\alpha }^{A}\left( z\right) :=K_{~\alpha }^{A}-zB_{~\alpha }^{A}=\left( 
\mathfrak{N}_{~\beta }^{Bb}\beta _{a}\right) -z\left( \mathfrak{N}_{~\beta
}^{Bb}\omega _{b}\right) :X\rightarrow E  \label{T_z}
\end{equation}

Note that we have suppressed subindex $x$ in vectorial spaces $X_{x}$ and $%
E_{x}.$

These operators $K_{~\alpha }^{A},B_{~\alpha }^{A}:X\rightarrow E$ take
elements $\phi ^{\alpha }$ in the vector space $X,$ with $\dim \left(
X\right) =u,$ and give elements $l^{A}$ in the vector space $E$, with $\dim
\left( E\right) =e.$ Because we are interested in systems with constraints,
we shall consider operators with $\dim \left( E\right) \geq \dim \left(
X\right) .$ From now on Greek indices go to $1,...,u$ and capital Latin to $%
1,...,e$. We call $X^{\prime }$ and $E^{\prime }$ to the dual spaces of $X$
and $E$ and $\phi _{\alpha }$ and $l_{A}$ to their elements, respectively$.$
We call right kernel of $T_{~\alpha }^{A}$ to the vectors $\phi ^{\alpha }$
such that $T_{~\alpha }^{A}\phi ^{\alpha }=0;$ and we call left kernel to
the covector $l_{A}$ such that $l_{A}T_{~\alpha }^{A}=0.$ We refer to $%
T_{~1}^{A},$ $T_{~2}^{A},...,$ $T_{~u}^{A}$ as the columns of $T_{~\alpha
}^{A}$ and $T_{~\alpha }^{1},T_{~\alpha }^{2},...,T_{~\alpha }^{3}$ as the
rows of $T_{~\alpha }^{A}$. We call $\sigma _{i}\left[ T_{~\alpha
}^{A}\left( z\right) \right] $ to the singular values of $T_{~\alpha
}^{A}\left( z\right) $ for any choice of the Hermitian forms $G_{1AB}$ and $%
G_{2\alpha \beta }$. Finally we use a bar $\bar{T}_{~\alpha }^{A}$ to denote
the complex conjugate of $T_{~\alpha }^{A}.$

The key idea of this section is to perturb the operator (\ref{T_z}) with
another appropriate operator (as in equation (\ref{eq_per_1})), linear in a
real, small, parameter factor $\varepsilon $,\ and study how the singular
values change. For that, we follow \cite{soderstrom1999perturbation}, \cite%
{moro2002first}, \cite{ji1988sensitivity}, \cite{stewart1984second}. In
particular S\"{o}derstr\"{o}m \cite{soderstrom1999perturbation} and Moro et.
al. \cite{moro2002first} show that the singular values have Taylor expansion
in $\left\vert \varepsilon \right\vert $, at least up to order two and this
will be crucial for the following results. They also give closed form
expressions for the first order term, using left and right eigenvectors.

Roughly speaking our first two results are for square operators. We shall
show that an operator is Jordan diagonalizable if and only if, each of their
perturbed singular values are order $O\left( \left\vert \varepsilon
\right\vert ^{0}\right) $ or $O\left( \left\vert \varepsilon \right\vert
^{1}\right) $. Thus we add another equivalent condition to Kreiss's Matrix
Theorem. In addition, we shall extend this result and show that: a perturb
singular value is order $O\left( \left\vert \varepsilon \right\vert
^{l}\right) $ if and only if the operator has an $l-$Jordan block\footnote{%
We called $l-$Jordan block to the matrix $J_{l}\left( \lambda \right)
=\left( 
\begin{array}{cccc}
\lambda  & 1 & 0 & 0 \\ 
0 & ... & ... & 0 \\ 
0 & 0 & \lambda  & 1 \\ 
0 & 0 & 0 & \lambda 
\end{array}%
\right) \in 
\mathbb{C}
^{l\times l}$ with eigenvalue $\lambda .$}, associated to some eigenvalue,
in the Jordan decomposition.

These results lead us to obtain a necessary condition for strong
hyperbolicity on rectangular operators. This is a necessary condition for
the existence of a reduction from rectangular to square operators, such
that, the reduced one is diagonalizable. The conclusion is analogous to the
square case, if any singular value of the perturbed operator is order $%
O\left( \left\vert \varepsilon \right\vert ^{l}\right) $ with $l\geq 2$,
then the system can not be reduced to a diagonalizable operator i.e. strong
hyperbolicity only admits orders $O\left( \left\vert \varepsilon \right\vert
^{0}\right) $ and $O\left( \left\vert \varepsilon \right\vert ^{1}\right) $.
Moreover, if the singular values are order $O\left( \left\vert \varepsilon
\right\vert ^{l}\right) $ then any reduction leads to operators with $l-$%
Jordan blocks or larger.

\subsection{Square operators. \label{square_operators_1}}

We consider first the space of equations such that $\dim \left( E\right)
=\dim \left( X\right) .$ For simplicity we shall identify $E$ with $X$, but
in general there is no natural identification between them. We also consider
a square operator $T_{~\beta }^{\alpha }\left( z\right) =K_{~\beta }^{\alpha
}-zB_{~\beta }^{\alpha }:X\rightarrow X$ with $z\in 
\mathbb{C}
$ and $B_{~\beta }^{\alpha }$ invertible (without right kernel). We call $%
\lambda _{i}$ $i=1,...,k$ the different eigenvalues of $\left( B^{-1}\right)
_{~\alpha }^{\gamma }K_{~\beta }^{\alpha }$;  $q_{i},$ $r_{i}$ their
respective geometric and algebraic multiplicities, and $D_{\lambda
}:=\left\{ \lambda _{i}\text{ with }i=1,...,k\right\} $.

In the following lemma, we shall use the SVD of $T_{~\beta }^{\alpha }\left(
z\right) $ and show for which $z$ the operator $T_{~\beta }^{\alpha }\left(
z\right) $ has vanishing singular values and how many there are.

\begin{lemma}
\label{Lemma_lam1}1) $T_{~\beta }^{\alpha }\left( \lambda _{i}\right) $ has
exactly $q_{i}$ null singular values. The rest of $u-q_{i}$ singular values
of $T_{~\beta }^{\alpha }\left( \lambda _{i}\right) $ are positive.

2) $\sigma _{i}\left[ T_{~\beta }^{\alpha }\left( z\right) \right] >0$ for
all singular values of $T_{~\beta }^{\alpha }\left( z\right) $ if and only
if $z\notin D_{\lambda }$

3) Consider any given subset $L\subset 
\mathbb{C}
,$ then 
\begin{equation*}
\sigma _{i}\left[ T_{~\beta }^{\alpha }\left( z\right) \right] >0\text{ }%
\forall z\in L\text{ and }\forall i=1,...,u
\end{equation*}

if and only if $D_{\lambda }\cap ~L=\phi $
\end{lemma}

\begin{proof}
1) Notice that \ $T_{~\beta }^{\alpha }\left( z\right) =B_{~\eta }^{\alpha
}\left( B^{-1}\right) _{~\gamma }^{\eta }T_{~\beta }^{\gamma }\left(
z\right) =B_{~\eta }^{\alpha }\left( \left( B^{-1}\right) _{~\gamma }^{\eta
}K_{~\beta }^{\gamma }-z\delta _{~\beta }^{\eta }\right) .$ It is clear from
this expression that the $right\_\ker \left( T\left( z\right) \right)
=right\_\ker \left( B^{-1}K-z\delta \right) $. Therefore $T_{~\beta
}^{\alpha }\left( z\right) $ has kernel only when $z$ is equal to one
eigenvalue of $B^{-1}\circ K.$

On the other hand, the singular value decomposition of $T$ is%
\begin{equation*}
T_{~\beta }^{\alpha }\left( z\right) =U_{~i^{\prime }}^{\alpha }\left(
z\right) \Sigma _{~j^{\prime }}^{i^{\prime }}\left( z\right) \left(
V^{-1}\right) _{~\beta }^{j^{\prime }}\left( z\right) 
\end{equation*}

Now $U,$ $\Sigma ,V^{-1}$ are operators that depend on $z$, and from the
orthogonality conditions (\ref{a_lem3_a}), (\ref{a_lem3_b}) in appendix \ref%
{appendix}, $U\left( z\right) $ and $V^{-1}\left( z\right) $ are always
invertible $\forall z\in 
\mathbb{C}
.$ Thus $\Sigma \left( z\right) $ is diagonal and controls the kernel of $T$
(this argument is valid for the rectangular case too). \ Consider now the
case $z=\lambda _{i}$ we know that $\dim \left( right\_\ker T\left( \lambda
_{i}\right) \right) =q_{i}$ but from Corollary \ref{Cor_ker} in appendix %
\ref{appendix}, it is the number of vanishing singular values.

2) and 3) are particular cases of 1).
\end{proof}

The operator $\left( B^{-1}\right) _{~\alpha }^{\gamma }K_{~\beta }^{\alpha
} $ is Jordan diagonalizable when $q_{i}=r_{i}$ $\forall i$ \ and from the
previous Lemma, this is only possible if the dimension of the right kernel%
\footnote{%
Or left kernel, since for square operators the dimension of right and left
kernels are equal.} of $T_{~\beta }^{\alpha }\left( \lambda _{i}\right) $ $%
\forall i$ is maximum. We shall see under which conditions this becomes
true. But first we need a previous Lemma.

Point 1 in the following Lemma is valid for rectangular operators too. We
shall use it also in section \ref{chap_examples} to give a condition for
hyperbolicity in the general case.

\begin{lemma}
\label{lemma_f1}1) Given $P:X\rightarrow E$ a linear rectangular operator
with $\dim \left( E\right) \geq \dim \left( X\right) $. Then 
\begin{equation}
\sqrt{\det \left( P^{\ast }\circ P\right) }=\dprod\limits_{i=1}^{u}\sigma
_{i}\left[ P\right]  \label{lem_f1_e1}
\end{equation}

2) Consider the square operator $T_{~\beta }^{\alpha }\left( z\right)
=K_{~\beta }^{\alpha }-zB_{~\beta }^{\alpha }:X\rightarrow X$. Then 
\begin{equation}
\sqrt{\det \left( T^{\ast }\circ T\right) }=\dprod\limits_{i=1}^{u}\sigma
_{i}\left[ B_{~\beta }^{\alpha }\right] \left\vert \lambda _{1}-z\right\vert
^{r_{1}}...\left\vert \lambda _{k}-z\right\vert
^{r_{k}}=\dprod\limits_{i=1}^{u}\sigma _{i}\left[ T_{~\beta }^{\alpha
}\left( z\right) \right]  \label{lem_sing_eingen}
\end{equation}
\end{lemma}

\begin{proof}
1) Consider the SVD of $P_{~\alpha }^{A}=\left( U_{P}\right) _{~i}^{A}\left(
\Sigma _{P}\right) _{~i^{\prime }}^{i}\left( V_{P}^{-1}\right) _{~\alpha
}^{i^{\prime }}$. Here $\left( U_{P}\right) _{~i}^{A}\in 
\mathbb{C}
^{e\times e}$, $\left( \Sigma _{P}\right) _{~i^{\prime }}^{i}\in 
\mathbb{R}
^{e\times u}$ and $\left( V_{P}^{-1}\right) _{~\alpha }^{i^{\prime }}\in 
\mathbb{C}
^{u\times u}$\ (see Theorem \ref{SVD_theorem}). Then 
\begin{eqnarray*}
\left( P^{\ast }\circ P\right) _{~\alpha }^{\beta } &=&G_{2}^{\beta \alpha
_{2}}\left( \bar{V}_{P}^{-1}\right) _{~\alpha _{2}}^{i_{2}^{\prime }}\left(
\Sigma _{P}\right) _{~i_{2}^{\prime }}^{j_{2}}\left( \bar{U}_{P}\right)
_{~j_{2}}^{A_{2}}G_{1A_{2}A_{1}}\left( U_{P}\right) _{~j_{1}}^{A_{1}}\left(
\Sigma _{P}\right) _{~i_{1}^{\prime }}^{j_{1}}\left( V_{P}^{-1}\right)
_{~\alpha }^{i_{1}^{\prime }} \\
&=&V_{~i_{1}^{\prime }}^{\beta }\delta _{2}^{i_{1}^{\prime }i_{2}^{\prime
}}\left( \Sigma _{P}\right) _{~i_{2}^{\prime }}^{j_{2}}\delta
_{1j_{2}j_{1}}\left( \Sigma _{P}\right) _{~i_{1}^{\prime }}^{j_{1}}\left(
V_{P}^{-1}\right) _{~\alpha }^{i_{1}^{\prime }}
\end{eqnarray*}%
where we have used the orthogonality conditions $\left( \bar{U}_{P}\right)
_{~i_{2}}^{A_{2}}G_{1A_{2}A_{1}}\left( U_{P}\right) _{~i_{1}}^{A_{1}}=\delta
_{i_{2}i_{1}}$ and\ $G_{2}^{\beta \alpha _{2}}\left( \bar{V}_{P}^{-1}\right)
_{~\alpha _{2}}^{i_{2}^{\prime }}=V_{~j_{2}^{\prime }}^{\beta }\delta
_{2}^{j_{2}^{\prime }i_{2}^{\prime }}$.

Taking determinant and square root 
\begin{eqnarray*}
\sqrt{\det \left( \left( P^{\ast }\circ P\right) _{~\alpha }^{\beta }\right) 
} &=&\sqrt{\det \left( V_{~i_{2}^{\prime }}^{\beta }\delta
_{2}^{i_{2}^{\prime }j_{2}^{\prime }}\left( \Sigma _{P}\right)
_{~j_{2}^{\prime }}^{i_{2}}\delta _{1i_{2}i_{1}}\left( \Sigma _{P}\right)
_{~j_{1}^{\prime }}^{i_{1}}\left( V_{P}^{-1}\right) _{~\alpha
}^{j_{1}^{\prime }}\right) } \\
&=&\sqrt{\det \left( \delta _{2}^{i_{2}^{\prime }j_{2}^{\prime }}\left(
\Sigma _{P}\right) _{~j_{2^{\prime }}}^{i_{2}}\delta _{i_{2}i_{1}}\left(
\Sigma _{P}\right) _{~j_{1}^{\prime }}^{i_{1}}\right) } \\
&=&\dprod\limits_{i=1}^{u}\sigma _{i}\left[ P_{~\alpha }^{A}\right]
\end{eqnarray*}

2) Similarly, taking the determinant of $T^{\ast }\circ T$ we get,%
\begin{eqnarray*}
\sqrt{\det \left( T^{\ast }\circ T\right) } &=&\sqrt{\det \left( G_{2}^{\eta
\beta }\bar{T}_{~\beta }^{\alpha }\left( z\right) G_{1\alpha \gamma
}T_{~\beta }^{\gamma }\left( z\right) \right) } \\
&=&\sqrt{\det \left( G_{2}^{\alpha \rho }\left( \left( \bar{B}^{-1}\right)
_{~\gamma }^{\mu }\bar{K}_{~\rho }^{\gamma }-\bar{z}\delta _{~\rho }^{\mu
}\right) \bar{B}_{~\mu }^{\upsilon }G_{1\upsilon \gamma }B_{~\eta }^{\gamma
}\left( \left( B^{-1}\right) _{~\gamma }^{\eta }K_{~\beta }^{\gamma
}-z\delta _{~\beta }^{\eta }\right) \right) } \\
&=&\sqrt{\det \left( G_{2}^{\alpha \mu }\bar{B}_{~\mu }^{\upsilon
}G_{1\upsilon \gamma }B_{~\eta }^{\gamma }\right) \det \left( \left( \bar{B}%
^{-1}\right) _{~\gamma }^{\mu }\bar{K}_{~\alpha }^{\gamma }-\bar{z}~\delta
_{~\alpha }^{\mu }\right) \det \left( \left( B^{-1}\right) _{~\gamma }^{\eta
}K_{~\beta }^{\gamma }-z\delta _{~\beta }^{\eta }\right) } \\
&=&\sqrt{\det \left( G_{2}^{\alpha \mu }\bar{B}_{~\mu }^{\upsilon
}G_{1\upsilon \gamma }B_{~\eta }^{\gamma }\right) }\left\vert \lambda
_{1}-z\right\vert ^{r_{1}}...\left\vert \lambda _{k}-z\right\vert ^{r_{k}} \\
&=&\sigma _{1}\left[ B\right] ...\sigma _{u}\left[ B\right] \left\vert
\lambda _{1}-z\right\vert ^{r_{1}}...\left\vert \lambda _{k}-z\right\vert
^{r_{k}}
\end{eqnarray*}%
In the fourth line we have used $\det \left( \left( B^{-1}\right) _{~\gamma
}^{\eta }K_{~\beta }^{\gamma }-z\delta _{~\beta }^{\eta }\right) =\left(
\lambda _{1}-z\right) ^{r_{1}}...\left( \lambda _{k}-z\right) ^{r_{k}}$ and
on the last line we have used the first point of the Lemma for $B.$
Therefore using it again for $T,$ we conclude 
\begin{equation*}
\left( \dprod\limits_{i=1}^{u}\sigma _{i}\left[ B\right] \right) \left\vert
\lambda _{1}-z\right\vert ^{r_{1}}...\left\vert \lambda _{k}-z\right\vert
^{r_{k}}=\sqrt{\det \left( T^{\ast }\circ T\right) }=\sigma _{1}\left[ T%
\right] ...\sigma _{u}\left[ T\right] 
\end{equation*}%
\bigskip 
\end{proof}

Notice that if in equation (\ref{lem_sing_eingen})\ we set $z=\lambda
_{1}+\varepsilon $ (with $\varepsilon $ real and small), then the product of
the singular values are order $O\left( \left\vert \varepsilon \right\vert
^{r_{1}}\right) $. Since these singular values have Taylor expansions in $%
\left\vert \varepsilon \right\vert $, if all singular values are $O\left(
\left\vert \varepsilon \right\vert ^{l}\right) $ with $l<2,$ then we need $%
r_{1}$ of them to vanish (that is $O\left( \left\vert \varepsilon
\right\vert ^{1}\right) $). Therefore by the previous Lemma $q_{1}=r_{1}.$
If this happens for all $\lambda _{i}$ then $q_{i}=r_{i}$ $\forall i$ and
the operator $\left( B^{-1}\right) _{~\alpha }^{\gamma }K_{~\beta }^{\alpha }
$ is Jordan diagonalizable.

A formalization of this idea is given in the next theorem. Notice that the
orders of the singular values are invariant under different choices of
Hermitian forms, although the singular values are not. We show this in
appendix \ref{appendix_2}.

\begin{theorem}
\label{Matrix_theorem}The following conditions are equivalent:

1) $\left( B^{-1}\right) _{~\alpha }^{\gamma }K_{~\beta }^{\alpha }$ is
Jordan diagonalizable.

2) $T_{~\beta }^{\alpha }\left( \lambda _{i}\right) =K_{~\beta }^{\alpha
}-\lambda _{i}B_{~\beta }^{\alpha }$ has $r_{i}$ vanishing singular values
for each $\lambda _{i}$.

3) For at least one fixed $\theta \in \left[ 0,2\pi \right] $ and $0\leq
\left\vert \varepsilon \right\vert <<1$ \ with $\varepsilon $\ real, the
singular values of the perturbed operators $T_{~\beta }^{\alpha }\left(
\lambda _{i}+\varepsilon e^{i\theta }\right) =T_{~\beta }^{\alpha }\left(
\lambda _{i}\right) -\varepsilon e^{i\theta }B_{~\beta }^{\alpha }$ are
either of two forms%
\begin{eqnarray}
\sigma _{j}\left[ T_{~\beta }^{\alpha }\left( \lambda _{i}+\varepsilon
e^{i\theta }\right) \right]  &=&\sigma _{j}\left[ T_{~\beta }^{\alpha
}\left( \lambda _{i}\right) \right] +\xi _{j}\varepsilon +O\left(
\varepsilon ^{2}\right) \text{ with }\sigma _{j}\left[ T_{~\beta }^{\alpha
}\left( \lambda _{i}\right) \right] \neq 0\text{ \ or}  \notag \\
\sigma _{j}\left[ T_{~\beta }^{\alpha }\left( \lambda _{i}+\varepsilon
e^{i\theta }\right) \right]  &=&\xi _{j}\left\vert \varepsilon \right\vert
+O\left( \left\vert \varepsilon \right\vert ^{2}\right) \text{ \ \ with}\
\xi _{j}\neq 0  \label{f_1}
\end{eqnarray}

\footnote{%
The orders in $\varepsilon $ are independent of $\theta .$} for all $\lambda
_{i}\in D_{\lambda }$ i.e. none of them is $\sigma \left[ T_{~\beta
}^{\alpha }\left( \lambda _{i}+\varepsilon e^{i\theta }\right) \right]
=O\left( \left\vert \varepsilon \right\vert ^{l}\right) $ with $l\geq 2.$
\end{theorem}

\begin{proof}
$1)\Longleftrightarrow 2)$ Since the geometric and algebraic multiplicities
are equal for all eigenvalues, i.e. $q_{i}=r_{i}$ $\forall i=1,...,k$.

$3)\Longleftrightarrow 1)$ Using Lemma \ref{lemma_f1} we have 
\begin{equation}
\left( \dprod\limits_{i=1}^{u}\sigma _{i}\left[ B\right] \right) \left\vert
\lambda _{1}-z\right\vert ^{r_{1}}...\left\vert \lambda _{k}-z\right\vert
^{r_{k}}=\sigma _{1}\left[ K-zB\right] ...\sigma _{u}\left[ K-zB\right]
\label{tf_5}
\end{equation}

Set $z=\lambda _{i}+\varepsilon e^{i\alpha }$ with $\varepsilon $ less than 
any distance between the eigenvalues to $\lambda _{i}$ 
\begin{equation}
\varepsilon <\min \left\{ \left\vert \lambda _{i}-\lambda _{j}\right\vert 
\text{ with }j=1,...,u\text{ \ and }i\neq j\right\}   \label{tf_2}
\end{equation}

By Lemma \ref{Lemma_lam1}, we know that $q_{i}$ singular values have to
vanish for $z=\lambda _{i}.$ Suppose they are the first $q_{i}$, we call
them $\left( \sigma _{\lambda _{i}}\right) _{j}\left[ K-zB\right] $, with $%
j=1,...,q_{i}$ then we rewrite equation (\ref{tf_5})%
\begin{equation}
\varepsilon ^{r_{i}}=\left( \sigma _{\lambda _{i}}\right) _{1}\left[ \left(
K-\lambda _{i}B\right) +\varepsilon \left( -e^{i\theta }B\right) \right]
...\left( \sigma _{\lambda _{i}}\right) _{q_{i}}\left[ \left( K-\lambda
_{i}B\right) +\varepsilon \left( -e^{i\theta }B\right) \right] ~\left.
p\left( z\right) \right\vert _{z=\lambda _{i}+\varepsilon e^{i\theta }}
\label{tf_3}
\end{equation}%
where

\small{
$\left. p\left( z\right) \right\vert _{z=\lambda _{i}+\varepsilon
e^{i\alpha }}=\left. \frac{\sigma _{q_{i}+1}\left( K-zB\right) ...\sigma
_{u}\left( K-zB\right) }{\left( \dprod\limits_{i=1}^{u}\sigma _{i}\left[ B%
\right] \right) \left\vert \lambda _{1}-\lambda _{i}-\varepsilon e^{i\alpha
}\right\vert ^{r_{1}}...\left\vert \lambda _{i-1}-\lambda _{i}-\varepsilon
e^{i\alpha }\right\vert ^{r_{i-1}}\left\vert \lambda _{i+1}-\lambda
_{i}-\varepsilon e^{i\alpha }\right\vert ^{r_{i+1}}...\left\vert \lambda
_{k}-\lambda _{i}-\varepsilon e^{i\alpha }\right\vert ^{r_{k}}}\right\vert
_{z=\lambda _{i}+\varepsilon e^{i\theta }}$}. Note that $\left. p\left(
z\right) \right\vert _{z=\lambda _{i}+\varepsilon e^{i\alpha }}$ does non
vanish for Lemma \ref{Lemma_lam1} and does not blow up for an $\varepsilon $
small enough because of equation (\ref{tf_2}).

We know for $2$ in theorem \ref{TTor}, in appendix \ref{appendix_2}, that
for $\left\vert \varepsilon \right\vert <<1$ the $\sigma ^{\prime }s$ can be
expanded as 
\begin{eqnarray*}
\left( \sigma _{\lambda _{i}}\right) _{j}\left[ K-\left( \lambda
_{i}+\varepsilon e^{i\theta }\right) B\right] &=&\left( \sigma _{\lambda
_{i}}\right) _{j}\left[ \left( K-\lambda _{i}B\right) +\varepsilon \left(
-e^{i\theta }B\right) \right] \\
&=&\left\vert \varepsilon \right\vert \xi _{j}+O\left( \varepsilon
^{2}\right)
\end{eqnarray*}

For some $\xi _{j}$ as in equation (\ref{TTor_1}). 

If we replace the last expression in (\ref{tf_3}) we obtain 
\begin{equation}
\left\vert \varepsilon \right\vert ^{r_{i}-q_{i}}=\xi _{1}...\xi
_{q_{i}}\left. p\left( z\right) \right\vert _{z=\lambda _{i}+\varepsilon
e^{i\alpha }}+O\left( \varepsilon \right)   \label{eq_proof_epsilon}
\end{equation}

Therefore:

$\blacklozenge $ $\ 3)\Longrightarrow 1)$ By hypothesis $\xi _{j}\neq 0$ for 
$j=1,...,q_{i}$; then equation (\ref{eq_proof_epsilon}) can only be valid if 
$q_{i}=r_{i}$ $\forall i=1,...,k$ (taking small enough $\varepsilon $).
Therefore $\left( B^{-1}\right) _{~\alpha }^{\gamma }K_{~\beta }^{\alpha }$
is diagonalizable.

$\blacklozenge $ $\ 1)\Longrightarrow 3)$ If $\left( B^{-1}\right) _{~\alpha
}^{\gamma }K_{~\beta }^{\alpha }$ is diagonalizable then $r_{i}=q_{i}$ and
taking $\varepsilon \rightarrow 0$ we obtain $1=\xi _{1}...\xi
_{q_{i}}\left. p\left( z\right) \right\vert _{z=\lambda _{1}}.$ Which
implies $\xi _{j}\neq 0$ for $j=1,...,q_{i}.$ Because $r_{i}=q_{i}$ \ for
all $i,$ we conclude the proof.
\end{proof}

An interpretation of condition $3$ in the above theorem is the following,
for any non Jordan diagonalizable square operator, you can always find a
right eigenvector, such that, the contraction of it with all left
eigenvectors vanishes. This is clearly impossible if the operator is
diagonalizable. We show this in the next example, consider the matrix 
\begin{equation*}
K_{~\beta }^{\alpha }=P_{~i^{\prime }}^{\alpha }\left( 
\begin{array}{cccc}
\lambda _{1} & 0 & 0 & 0 \\ 
0 & \lambda _{1} & 1 & 0 \\ 
0 & 0 & \lambda _{1} & 1 \\ 
0 & 0 & 0 & \lambda _{1}%
\end{array}%
\right) _{~j^{\prime }}^{i^{\prime }}\left( P^{-1}\right) _{~\beta
}^{j^{\prime }}
\end{equation*}

and $B_{~\beta }^{\alpha }=\delta _{~\beta }^{\alpha }$ the identity matrix.
We call $v_{1,2}$ to the right eigenvectors and $u_{1,2}$ to the left
eigenvectors%
\begin{equation*}
\begin{array}{ccc}
\left( v_{1}\right) ^{\alpha }=P_{~i^{\prime }}^{\alpha }\left( 
\begin{array}{c}
1 \\ 
0 \\ 
0 \\ 
0%
\end{array}%
\right) ^{i^{\prime }} &  & \left( v_{2}\right) ^{\alpha }=P_{~i^{\prime
}}^{\alpha }\left( 
\begin{array}{c}
0 \\ 
1 \\ 
0 \\ 
0%
\end{array}%
\right) ^{i^{\prime }}%
\end{array}%
\end{equation*}%
\begin{equation*}
\begin{array}{ccc}
\left( u_{1}\right) _{\alpha }=\left( 
\begin{array}{cccc}
1 & 0 & 0 & 0%
\end{array}%
\right) _{j^{\prime }}\left( P^{-1}\right) _{~\alpha }^{j^{\prime }} &  & 
\left( u_{2}\right) _{~\alpha }=\left( 
\begin{array}{cccc}
0 & 0 & 0 & 1%
\end{array}%
\right) _{j^{\prime }}\left( P^{-1}\right) _{~\alpha }^{j^{\prime }}%
\end{array}%
\end{equation*}

then 
\begin{equation}
\left( u_{1,2}\right) _{\alpha }\left( v_{2}\right) ^{\alpha }=0
\label{contract_ker}
\end{equation}

Theorem \ref{TTor}, in appendix \ref{appendix_2}, tells us how to calculate
the coefficients of first order perturbation of the singular values of $%
T_{~\beta }^{\alpha }\left( \lambda _{1}+\varepsilon e^{i\theta }\right)
=\left( K_{~\beta }^{\alpha }-\lambda _{1}\delta _{~\beta }^{\alpha }\right)
-\varepsilon e^{i\theta }\delta _{~\beta }^{\alpha }$. They are given by $%
\xi _{i}=\sigma _{i}\left[ L_{~k}^{j}\right] $ with $L_{~k}^{j}:=\left( 
\begin{array}{c}
\left( u_{1}\right) _{\alpha } \\ 
\left( u_{2}\right) _{\alpha }%
\end{array}%
\right) \delta _{~\beta }^{\alpha }\left( \left( v_{1}\right) ^{\beta
},\left( v_{2}\right) ^{\beta }\right) $. But because of equation (\ref%
{contract_ker}), $L_{~k}^{j}$ has kernel $\left( 
\begin{array}{c}
0 \\ 
1%
\end{array}%
\right) $ and so $\xi _{2}=0$. Thus $T_{~\beta }^{\alpha }\left( \lambda
_{1}+\varepsilon e^{i\theta }\right) $ has a singular value of order $%
O\left( \varepsilon ^{2}\right) $ and it is not diagonalizable.

Let us go back to the general case. When the eigenvalues of $\left(
B^{-1}\right) _{~\alpha }^{\gamma }K_{~\beta }^{\alpha }$ are real, then the
below Corollary follows. This is equivalent to Theorem \ref{Teo_no_jordan}.

\begin{corollary}
\label{cor_fk}The next conditions are equivalent

1) $\left( B^{-1}\right) _{~\alpha }^{\gamma }K_{~\beta }^{\alpha }$ is
Jordan diagonalizable with real eigenvalues

2) All singular values satisfy 
\begin{equation}
\sigma _{j}\left[ T_{~\beta }^{\alpha }\left( x+iy\right) \right] >0\text{
with }x,y\in 
\mathbb{R}
\text{ and }y\neq 0  \label{cor_fk_2}
\end{equation}

For at least one fixed $\theta \in \left[ 0,2\pi \right] $ and$\ 0\leq
\left\vert \varepsilon \right\vert <<1$ with $\varepsilon $ real,%
\begin{eqnarray}
\sigma _{j}\left[ T_{~\beta }^{\alpha }\left( x+i\varepsilon e^{i\theta
}\right) \right]  &=&\sigma _{j}\left[ T_{~\beta }^{\alpha }\left( x\right) %
\right] +\xi _{j}\varepsilon +O\left( \varepsilon ^{2}\right) \text{ with }%
\sigma _{j}\left[ T_{~\beta }^{\alpha }\left( x\right) \right] \neq 0\text{
\ or}  \notag \\
\sigma _{j}\left[ T_{~\beta }^{\alpha }\left( x+i\varepsilon e^{i\theta
}\right) \right]  &=&\xi _{j}\left\vert \varepsilon \right\vert +O\left(
\varepsilon ^{2}\right) \text{ \ with }\xi _{j}\neq 0  \label{cor_fk_1}
\end{eqnarray}

for any $x\in 
\mathbb{R}
$ i.e. none of them is $\sigma \left[ T_{~\beta }^{\alpha }\left(
x+i\varepsilon e^{i\theta }\right) \right] =O\left( \left\vert \varepsilon
\right\vert ^{l}\right) $ with $l\geq 2$.
\end{corollary}

\begin{proof}
$1)\Longrightarrow 2).$ It follows directly from Theorem \ref{Matrix_theorem}%
.

$2)$ $\Longrightarrow 1).$ Because $\sigma _{i}\left[ T_{~\beta }^{\alpha
}\left( z\right) \right] >0$ $\forall i$ and $\forall z\in S$ with $%
S=\left\{ z\in 
\mathbb{C}
\text{ }/\text{ }\func{Im}\left( z\right) \neq 0\right\} ,$ then from Lemma %
\ref{Lemma_lam1}, $S\cap D_{\lambda }=\phi .$ Therefore the eigenvalues are
real. The second part also follows form Theorem \ref{Matrix_theorem}.
\end{proof}

An alternative proof to Theorem \ref{Matrix_theorem} can be obtained
directly showing that condition 2 in Corollary \ref{cor_fk}, with $\theta =%
\frac{\pi }{2},$ is equivalent to one of the conditions in Kreiss's Matrix
Theorem. Indeed, an alternative formulation is given by

\begin{theorem}[Part of Kreiss matrix Theorem]
\label{Kreiss_theorem}The square operator $K:V\rightarrow V$ is Jordan
diagonalizable with real eigenvalues if and only if for any $x,y\in 
\mathbb{R}
$ with $y\neq 0$  a constant $C\neq 0$ exists such that 
\begin{equation}
\left\Vert \left( K_{~\beta }^{\alpha }-\left( x+iy\right) \delta _{~\beta
}^{\alpha }\right) ^{-1}\right\Vert _{2}\leq \frac{C}{\left\vert
y\right\vert }  \label{k_1}
\end{equation}%
\bigskip 
\end{theorem}

We now show that (\ref{k_1}) implies (\ref{kr_1}) and this is equivalent to (%
\ref{cor_fk_1}). From \cite{stewart1990matrix} 
\begin{equation*}
\left\Vert T\right\Vert _{2}=\max \left\{ \sigma _{i}\left[ T\right] \right\}
\end{equation*}

where $\max \left\{ \sigma _{i}\left[ T\right] \right\} $ is the maximum of
all singular values of $T.$

In addition 
\begin{equation*}
\left\Vert T^{-1}\right\Vert _{2}=\max \left\{ \sigma _{i}\left[ T^{-1}%
\right] \right\} =\max \left\{ \frac{1}{\sigma _{i}\left[ T\right] }\right\}
=\frac{1}{\min \left\{ \sigma _{i}\left[ T\right] \right\} }
\end{equation*}

Where we have used that the singular values of $T$ are the inverse\footnote{%
Notice that if the SVD of $T$ is $T=U\Sigma V^{-1}$ then $T^{-1}=V\Sigma
^{-1}U^{-1}$ because $V$ and $U^{-1}$ are orthogonal and $\Sigma ^{-1}$ is
diagonal, that is the SVD of $T^{-1}.$ Therefore $\sigma _{i}\left[ T^{-1}%
\right] =\frac{1}{\sigma _{i}\left[ T\right] }$} of the singular values of $%
T^{-1}$.

Now, from inequality (\ref{k_1}) \ 
\begin{equation*}
\frac{1}{\min \left\{ \sigma _{i}\left[ K_{~\beta }^{\alpha }-\left(
x+iy\right) \delta _{~\beta }^{\alpha }\right] \right\} }=\left\Vert \left(
K_{~\beta }^{\alpha }-\left( x+iy\right) \delta _{~\beta }^{\alpha }\right)
^{-1}\right\Vert _{2}\leq \frac{C}{\left\vert y\right\vert }
\end{equation*}

let $\tilde{C}:=\frac{1}{C},$ then the Kreiss's Matrix Theorem asserts that $%
K$ is Jordan diagonalizable if and only if we can find $\tilde{C}$ such that 
\begin{equation}
\tilde{C}\left\vert y\right\vert \leq \min \left\{ \sigma _{i}\left[
K_{~\beta }^{\alpha }-\left( x+iy\right) \delta _{~\beta }^{\alpha }\right]
\right\}   \label{kr_1}
\end{equation}

This equation is equivalent to condition (\ref{cor_fk_1}). Since $\tilde{C}%
\left\vert y\right\vert \leq \left\vert y\right\vert ^{l}$ with $l\geq 2$ \
in $0\leq \left\vert y\right\vert <<1$ implies that $\tilde{C}=0,$ therefore 
$\min \left\{ \sigma _{i}\left[ K_{~\beta }^{\alpha }-\left( x+iy\right)
\delta _{~\beta }^{\alpha }\right] \right\} $ must be order $O\left(
\left\vert y\right\vert ^{0}\right) $ or $O\left( \left\vert y\right\vert
^{1}\right) $. In addition, notice that $K_{~\beta }^{\alpha }-\left(
x+iy\right) \delta _{~\beta }^{\alpha }$, in (\ref{k_1}), is invertible when 
$x+iy$ is not an eigenvalue of $K_{~\beta }^{\alpha }$. Since condition (\ref%
{k_1}) applies for all $x+iy$ with $y\neq 0$ then the eigenvalues of $%
K_{~\beta }^{\alpha }$ are real, as equation (\ref{cor_fk_2}) implies.

As a side remark,\ we prove the following theorem.

\begin{theorem}
\label{deg_eigen}$T_{~\beta }^{\alpha }\left( \lambda _{i}+\varepsilon
e^{i\theta }\right) :X\rightarrow X$ has a singular value of order $O\left(
\varepsilon ^{l}\right) $ if and only if $\left( B^{-1}\right) _{~\alpha
}^{\gamma }K_{~\beta }^{\alpha }$ has a $l-$Jordan block, with eigenvalue $%
\lambda _{i}$ in its Jordan decomposition.
\end{theorem}

\begin{proof}
We consider a basis in which $\left( B^{-1}\right) _{~\alpha }^{\gamma
}K_{~\beta }^{\alpha }$ stays in its Jordan form. In this basis we choose
the following two Hermitian forms $G_{1AB}=diag\left( 1,...,1\right) $ and $%
G_{2\alpha \beta }=diag\left( 1,...,1\right) $. Then, the calculation of $%
\sigma _{i}\left[ T_{~\beta }^{\alpha }\left( \lambda \right) \right] $ $%
i=1,..,u$ decouples in Jordan blocks. Therefore we only need to study the
singular values of an $l-$Jordan block $J_{l}\left( \lambda \right) $. It is
easy to see from equation (\ref{lem_sing_eingen}), with $z=\lambda
+e^{i\theta }\varepsilon $, that $J_{l}\left( \lambda \right) $ has a unique
singular value of order $O\left( \left\vert \varepsilon \right\vert
^{l}\right) $ and the others are order $O\left( \left\vert \varepsilon
\right\vert ^{0}\right) $. This concludes the proof.
\end{proof}

\subsection{Rectangular operators.\label{Rectangular_operators}}

In this subsection we consider the case $\dim E\geq \dim X.$ Theorem \ref%
{Matrix_theorem_2} provides a necessary condition for reducing a rectangular
operator to a square one, in such a way that the resulting operator is
Jordan diagonalizable. The proof of condition 2 in Theorem \ref{sht_Theo_1}
is a Corollary of this Theorem. A reduction is given explicitly by another
linear operator $h_{~A}^{\alpha }:E\rightarrow X$ \ (see section \ref{cap2}%
). It selects some evolution equations from the space of equations in a
physical theory. The theorem asserts under which conditions it will be
impossible to find a hyperbolizer, namely, a reduction satisfying condition (%
\ref{ker_1}) for strong hyperbolicity.

The proof of this theorem is based in the following Lemma (for a proof see 
\cite{stewart1990matrix}).

\begin{lemma}
\label{Lemma_des}Consider the linear operators $T_{~\alpha
}^{A}:X\rightarrow E$, $h_{~A}^{a}:E\rightarrow X$ and $H_{~\beta }^{\alpha
}:=h_{~A}^{\alpha }T_{~\beta }^{A}:X\rightarrow X$ then 
\begin{equation}
0\leq \sigma _{i}[H_{~\beta }^{\alpha }]\leq \sigma _{i}[T_{~\beta
}^{A}]\max \left\{ \sigma _{j}\left[ h_{~A}^{\alpha }\right] \right\} 
\label{Lemma_des_eq1}
\end{equation}

\footnote{%
The singular values are $\sigma _{i}[H_{~\beta }^{\alpha }]=\sqrt{\lambda
_{i}\left[ G_{2}^{-1}\circ \bar{H}^{\prime }\circ G_{2}\circ H\right] }$, $%
\sigma _{i}[T_{~\beta }^{A}]=\sqrt{\lambda _{i}\left[ G_{2}^{-1}\circ \bar{T}%
^{\prime }\circ G_{1}\circ T\right] }$ and $\sigma _{i}[h_{~A}^{a}]=\sqrt{%
\lambda _{i}\left[ G_{1}^{-1}\circ \bar{h}^{\prime }\circ G_{2}\circ h\right]
}.$ Where $\lambda _{i}\left[ K\right] $ mean eigenvalues of $K.$}Where the
singular values have been ordered from larger to smaller for each
operator.\bigskip 
\end{lemma}

Consider for contradiction, that there exists a hyperbolizer. Namely a
surjective reduction $h_{~A}^{\alpha }:E\rightarrow X,$ \ that does not
depend on $z$, of the operator $T_{~\alpha }^{A}\left( z\right) =K_{~\alpha
}^{A}-zB_{~\alpha }^{A}:X\rightarrow E,$ in which $B_{~\alpha }^{A}$ has no
right kernel\footnote{%
Notice that as in the square case, a rectangular operator $T_{~\alpha }^{A}$
has right kernel when at least one of their singular values vanishes (see
proof of Lemma \ref{Lemma_lam1}). But this is equivalent to the vanishing of
equation (\ref{lem_f1_e1}). Therefore, $B_{~\alpha }^{A}$ has no right
kernel if and only if $\det \left( B^{\ast }\circ B\right) =\sigma _{1}\left[
B^{A}{}_{~\beta }\right] ...\sigma _{u}\left[ B_{~\beta }^{A}\right] \neq 0.$%
}, and such that $h_{~A}^{\alpha }B_{~\beta }^{A}$ is invertible.

Then the next Theorem follows

\begin{theorem}
\label{Matrix_theorem_2} Suppose that for at least one singular value of $%
T_{~\alpha }^{A}\left( \lambda +\varepsilon e^{i\theta }\right) $, with $%
\lambda \in 
\mathbb{C}
$, satisfies 
\begin{equation*}
\sigma \left[ T_{~\alpha }^{A}\left( \lambda +\varepsilon e^{i\theta
}\right) \right] =O\left( \varepsilon ^{l}\right) \text{ with }l\geq 2
\end{equation*}

then there exists at least one singular value of $h_{~A}^{\alpha }T_{~\alpha
}^{A}\left( \lambda +\varepsilon e^{i\theta }\right) $ such that 
\begin{equation*}
\sigma \left[ h_{~A}^{\alpha }T_{~\beta }^{A}\left( \lambda +\varepsilon
e^{i\theta }\right) \right] =O\left( \varepsilon ^{m}\right) \text{ \ \ \
with }m\geq l\geq 2
\end{equation*}

Thus $\left( \left( h_{~C}^{\alpha }B_{~\gamma }^{C}\right) ^{-1}\right)
_{~\gamma }^{\alpha }h_{~A}^{\gamma }K_{~\beta }^{A}$ is non-diagonalizable,
in particular there does not exists any hyperbolizer and system (\ref{sht_1}%
) is not strongly hyperbolic.
\end{theorem}

\begin{proof}
We use lemma \ref{Lemma_des} for $T_{~\beta }^{A}\left( \lambda +\varepsilon
e^{i\theta }\right) $ and $h_{~A}^{\alpha }T_{~\beta }^{A}\left( \lambda
+\varepsilon e^{i\theta }\right) $, and let $\frac{1}{C}:=\max \left\{
\sigma _{j}\left[ h_{~A}^{\alpha }\right] \right\} $ (It does not vanish
since $h_{~A}^{\alpha }\neq 0$ and does not depend on $\lambda $), then for
equation (\ref{Lemma_des_eq1}) 
\begin{equation*}
0\leq \sigma _{i}[h_{~A}^{\alpha }T_{~\beta }^{A}\left( \lambda +\varepsilon
e^{i\theta }\right) ]\leq \sigma _{i}[T_{~\beta }^{A}\left( \lambda
+\varepsilon e^{i\theta }\right) ]\frac{1}{C}
\end{equation*}%
But for some $i$ , $\sigma _{i}[T_{~\beta }^{A}\left( \lambda +\varepsilon
e^{i\theta }\right) ]=O\left( \varepsilon ^{l}\right) $ with $l\geq 2.$
Therefore $\sigma _{i}[h_{~A}^{\alpha }T_{~\beta }^{A}\left( \lambda
+\varepsilon e^{i\theta }\right) ]=O\left( \varepsilon ^{m}\right) $ with $%
m\geq l\geq 2$. Since $C\varepsilon ^{m}\leq \varepsilon ^{l}$ for $0\leq
\varepsilon <<1$ is only possible if $m\geq l.$

Applying Theorem \ref{Matrix_theorem} to 
\begin{equation*}
\tilde{T}_{~\beta }^{\alpha }=h_{~A}^{\alpha }K_{~\beta }^{A}-z\left(
h_{~A}^{\alpha }B_{~\beta }^{A}\right) 
\end{equation*}%
and recalling that $h_{~A}^{\alpha }B_{~\beta }^{A}$ is invertible by
hypothesis, we conclude that $\left( \left( h_{~C}^{\alpha }B_{~\gamma
}^{C}\right) ^{-1}\right) h_{~A}^{\gamma }K_{~\beta }^{A}$ is not
diagonalizable. \ Therefore it is not a hyperbolizer and we reach a
contradiction.

This result considers perturbation of the singular values around their
vanishing values. As it has been shown in appendix \ref{appendix_2}, the
orders of perturbations are invariant under any choice of these Hermitian
forms. Thus the result does not depend on the particularities of the SVD.
\end{proof}

\bigskip


\section{APPLICATIONS AND EXAMPLES\label{chap_examples}.}

In this section we shall show how to check conditions 1 and 2 in theorems %
\ref{Teo_no_jordan} and \ref{sht_Theo_1}. They are very simple to verify in
examples.

Condition 1: We shall assume that there exists $\omega _{a}$ such that $%
\mathfrak{N}_{~\gamma }^{Bb}\omega _{b}$ has no right kernel.

As we mentioned before, $\mathfrak{N}_{~\gamma }^{Bb}n\left( \lambda \right)
_{b}$ with $n\left( \lambda \right) _{b}\in S_{\omega _{a}}^{%
\mathbb{C}
}$, has right kernel when at least one of its singular values vanishes. This
happens if and only if, given any positive definite Hermitian forms $G_{1}$
and $G_{2}$, 
\begin{equation}
\sqrt{\det \left( G_{2}^{\alpha \gamma }\overline{\mathfrak{N}_{~\gamma
}^{Aa}n\left( \lambda \right) _{a}}G_{1AB}\mathfrak{N}_{~\beta }^{Bb}n\left(
\lambda \right) _{b}\right) }=0,  \label{sht_6}
\end{equation}

as it has been proved in Lemma \ref{lemma_f1}. Therefore, the system is
hyperbolic if and only if all roots $\lambda _{k}$ of this equation are
real. In addition, for any line $n\left( \lambda \right) _{b}$ we call
characteristic eigenvalues to their corresponding $\left\{ \lambda
_{k}\right\} $.

Condition 2: In general it is not an easy task to calculate the singular
values and their orders in parameter $\varepsilon $. Fortunately theorem \ref%
{Conjetura_final} below allows for a simpler calculation, showing when the
coefficient of zero and first order of the singular values vanish.

Assuming that condition 1 has been checked for $\omega _{a}$. Consider the
line $n_{a}\left( \lambda \right) =-\lambda \omega _{a}+\beta _{a}$
belonging to $S_{\omega _{a}}$, with $\beta _{a}$ not proportional to $%
\omega _{a}$ and $\lambda $ real. Let $\left\{ \lambda _{k}\right\} $ be the
characteristic eigenvalues of $n_{a}\left( \lambda \right) $. Then, the
principal symbol $\mathfrak{N}_{~\gamma }^{Aa}n_{a}\left( \lambda
_{k}\right) $ has right and left kernel. We call $W_{~i}^{\gamma }\left(
\lambda _{k}\right) $ and $U_{A}^{j}\left( \lambda _{k}\right) $ with $%
i=1,...,\dim \left( left\_\ker \left( \mathfrak{N}_{~\gamma
}^{Aa}n_{a}\left( \lambda _{k}\right) \right) \right) $ and\ $\ j=1,...,\dim
\left( right\_\ker \left( \mathfrak{N}_{~\gamma }^{Aa}n_{a}\left( \lambda
_{k}\right) \right) \right) $ to any basis of these spaces respectively,
namely they are linearly independent sets of vectors such that $\mathfrak{N}%
_{~\gamma }^{Aa}n_{a}\left( \lambda _{k}\right) W_{~i}^{\gamma }\left(
\lambda _{k}\right) =0$ and $U_{A}^{j}\left( \lambda _{k}\right) \mathfrak{N}%
_{~\gamma }^{Aa}n_{a}\left( \lambda _{k}\right) =0$.

Now consider a perturbation of these covectors $n_{\varepsilon ,\theta
}\left( \lambda _{i}\right) _{a}=-\varepsilon e^{i\theta }\omega
_{a}+n_{a}\left( \lambda _{i}\right) $ with $0\leq \left\vert \varepsilon
\right\vert <<1$ and $\varepsilon $ real, and any fixed $\theta \in \left[
0,2\pi \right] $.

\begin{theorem}
\label{Conjetura_final} A necessary condition for system (\ref{sht_1}) to be
strongly hyperbolic is: The following operator 
\begin{equation}
L_{~i}^{j}\left( \lambda _{k}\right) :=U_{A}^{j}\left( \lambda _{k}\right)
\left( \mathfrak{N}_{~\gamma }^{Aa}\omega _{a}\right) W_{~i}^{\gamma }\left(
\lambda _{k}\right)   \label{conj_L}
\end{equation}

has no right kernel\footnote{$L_{~i}^{j}$ has no kernel if and only if\
given any positive define Hermitian forms $G_{3}$ and $G_{4}$ then 
\begin{equation*}
\det \left( G_{3}^{ji_{1}}\overline{L_{~i_{1}}^{j_{1}}}%
G_{4j_{1}j_{2}}L_{~i}^{j_{2}}\right) \neq 0
\end{equation*}%
\par
{}}.
\end{theorem}

Definition of $L_{~i}^{j}\left( \lambda _{k}\right) $ is equivalent to $%
\tilde{L}_{~i}^{j}=\delta _{1}^{jj}\left( 0,\bar{U}_{1},\bar{U}_{3}\right)
_{~j}^{C}\delta _{1CD}B_{~\alpha }^{D}\left( 0,V_{1}\right) _{~i}^{\alpha }$%
, in equation \ref{TTor_1} in appendix \ref{appendix_2}, under a basis
transformation. If $\tilde{L}_{~i}^{j}$ has right kernel then it has a
singular value which vanishes, and then at least one perturbed singular
value $\sigma \left( \mathfrak{N}_{~\gamma }^{Aa}n_{\varepsilon ,\theta
}\left( \lambda _{i}\right) _{a}\right) $ is order $O\left( \left\vert
\varepsilon \right\vert ^{l}\right) $ with $l\geq 2$.

As we have shown in theorem \ref{Teo_no_jordan} for the square case, this is
also a sufficient condition.

Now, using the tools developed, we show how to apply these results in some
examples.

\subsection{Matrix example 1. \label{Example_matrix}}

Consider the matrix 
\begin{equation}
T\left( x\right) =K-zB:=\left( 
\begin{array}{cc}
\lambda _{1} & \kappa \\ 
0 & \lambda _{2}%
\end{array}%
\right) -z\left( 
\begin{array}{cc}
1 & 0 \\ 
0 & 1%
\end{array}%
\right) \in 
\mathbb{C}
^{2\times 2}  \label{me_1}
\end{equation}%
in which $\lambda _{1},\lambda _{2},\kappa $ are constants. Consider the
scalar products $G_{1,2}=\delta _{1,2}\,=diag\left( 1,1\right) ,\ $then the
singular values of $T\left( z\right) $ are 
\begin{eqnarray*}
\sigma _{1}\left[ T\left( z\right) \right] &=&\sqrt{\omega \left( z\right) +%
\sqrt{\omega ^{2}\left( x\right) -\left\vert z-\lambda _{1}\right\vert
^{2}\left\vert z-\lambda _{2}\right\vert ^{2}}} \\
\sigma _{2}\left[ T\left( z\right) \right] &=&\sqrt{\omega \left( z\right) -%
\sqrt{\omega ^{2}\left( x\right) -\left\vert z-\lambda _{1}\right\vert
^{2}\left\vert z-\lambda _{2}\right\vert ^{2}}}
\end{eqnarray*}

with $\ $ 
\begin{equation*}
\omega \left( z\right) =\frac{1}{2}\left( \left\vert z-\lambda
_{1}\right\vert ^{2}+\left\vert z-\lambda _{2}\right\vert ^{2}+\left\vert
\kappa \right\vert ^{2}\right)
\end{equation*}

a non-negative function of $z.$ Notice that $\sigma _{1}\left[ T\left(
z\right) \right] $ can be vanished only when $\lambda _{1}=\lambda _{2}$ and 
$\kappa =0.$

The Taylor expansion of $\sigma _{2}$ centered in $\lambda _{1,2}$ is%
\footnote{%
For calculate the Taylor expansion we use the identity$\ \ X-\sqrt{%
X^{2}-Y^{2}}=\frac{1}{2}\left( \sqrt{X+Y}-\sqrt{X-Y}\right) ^{2}$ with real $%
X,Y$ and\ \ $X+Y,X-Y>0$}:%
\begin{equation*}
\sigma _{2}\left( \lambda _{1,2}+\varepsilon \right) \approx 0+\frac{%
\left\vert \lambda _{1}-\lambda _{2}\right\vert }{\sqrt{\left\vert \lambda
_{1}-\lambda _{2}\right\vert ^{2}+\left\vert \kappa \right\vert ^{2}}}%
\left\vert \varepsilon \right\vert +O\left( \varepsilon ^{2}\right) 
\end{equation*}

As in theorem \ref{Teo_no_jordan}

\begin{itemize}
\item $K$ is not diagonalizable when $\lambda _{1}=\lambda _{2}$ and $\kappa
\neq 0$. In that case $\sigma _{2}\left( \lambda _{1,2}+\varepsilon \right)
=O\left( \left\vert \varepsilon \right\vert ^{2}\right) $.

\item $K$ is diagonalizable for any other case and the singular values
remain of order $O\left( \left\vert \varepsilon \right\vert \right) $ or $%
O\left( \varepsilon ^{0}\right) .$
\end{itemize}

\subsection{Force-Free Electrodynamics in Euler potential description. \label%
{Euler_potenciales}}

In this subsection we study a description of the Force-Free Electrodynamics
system based on Euler's potentials \cite{hawking1979general}, \cite%
{uchida1997theory}. When it is written as a first-order system, this is a
constrained system and we shall show that it is only weakly hyperbolic. It
is important to mention that Reula and Rubio \cite{reula2016ill} reached the
same conclusion by another method. They used the potentials as fields
obtaining a second-order system in derivatives, which then led to a
pseudodifferential first order system without constraints and finally tested
the failure of strong hyperbolicity using Kreiss criteria \cite%
{kreiss2004initial}. The advantage of our technique is that we use the
gradients of the potentials as fields, obtaining directly a first-order
system in partial derivatives but with constraints. Then, proving that
condition 2 in Theorem \ref{sht_Theo_1} fails, we conclude that there does
not exist any hyperbolizer.

In this system the electromagnetic tensor $F_{ab}$ is degenerated $%
F_{ab}j^{b}=0$ and magnetic dominated $F:=F_{ab}F^{ab}>0.$ These conditions
allows us to decompose $F_{ab}=l_{1[a}l_{2b]}$, (see, \cite%
{penrose1984spinors}, \cite{gralla2014spacetime}) in terms of space-like
1-forms $l_{ia}$ with $i=1,2$. For more detailed works on Force-Free
electrodynamics see \cite{komissarov2002time}, \cite%
{blandford1977electromagnetic}, \cite{carrasco2016covariant}, \cite%
{gralla2014spacetime}.

In addition, Carter 1979 \cite{hawking1979general} and Uchida 1997 \cite%
{uchida1997theory} proved that there exist two Euler potentials $\phi _{1}$
and $\phi _{2}$ such that $l_{ia}=\nabla _{a}\phi _{i}.$

With this ansatz, the Force Free equations in the (gradient) Euler's
potentials version are 
\begin{eqnarray*}
l_{ka}\nabla _{b}\left( l_{i}^{a}l_{j}^{b}\varepsilon ^{ij}\right) &=&0 \\
\nabla _{\lbrack a}l_{\left\vert i\right\vert b]} &=&0
\end{eqnarray*}

with background metric $g_{ab}.$ Taking a linearized version at a given
point and background solution, we get the following principal symbol 
\begin{equation*}
\mathfrak{N}_{~\alpha }^{Aa}\left( x,\phi \right) n_{a}\delta \phi ^{\alpha
}=\left( 
\begin{array}{cc}
\left( l_{1a}\left( l_{2}.n\right) -\left( l_{1.}l_{2}\right) n_{a}\right) & 
\left( \left( l_{1}.l_{1}\right) n_{a}-l_{1a}\left( l_{1}.n\right) \right)
\\ 
\left( l_{2a}\left( l_{2}.n\right) -\left( l_{2.}l_{2}\right) n_{a}\right) & 
\left( \left( l_{1.}l_{2}\right) n_{a}-l_{2a}\left( l_{1}.n\right) \right)
\\ 
n^{[b}\delta _{a}^{c]} & 0 \\ 
0 & n^{[b}\delta _{a}^{c]}%
\end{array}%
\right) \left( 
\begin{array}{c}
\delta l_{1}^{a} \\ 
\delta l_{2}^{a}%
\end{array}%
\right)
\end{equation*}

The solution space $\delta \phi ^{\alpha }=\left( 
\begin{array}{c}
\delta l_{1}^{a} \\ 
\delta l_{2}^{a}%
\end{array}%
\right) $ is $8$-dimensional and the associated space of equations is $14$%
-dimensional $\delta X_{A}=\left( \delta W,\delta X,\delta Y_{bc},\delta
Z_{bc}\right) $ where $\delta Y_{bc}=\delta Y_{[bc]}$ and $\delta
Z_{bc}=\delta Z_{[bc]}$.

1) We shall check that the system is hyperbolic: Consider $\omega _{a}$
timelike and normalized $\omega _{a}\omega ^{a}=-1$, since $l_{ia}$ can be
chosen orthogonal (via a gauge transformation), we define an orthonormal
frame $\left\{ e_{ia}\text{ }i=0,1,2,3\right\} $ with $e_{0a}=\omega _{a}$
and $l_{ia}=l_{i}e_{ia}$ with $i=1,2$ such that $g_{ab}=\left(
-1,1,1,1\right) $. Consider now the plane $n_{a}\left( \lambda \right)
=-n_{0}\omega _{a}+\beta _{a}\in S_{\omega _{a}}^{%
\mathbb{C}
}$ with $n_{0}\in 
\mathbb{C}
$, $\beta _{a}=n_{i}e_{ia}$ for $i=1,2,3$, $n_{i}$ real and let $%
G_{1AB}=diag\left( 1,...,1\right) $ and $G_{2}^{\alpha \beta }=diag\left(
1,...,1\right) $, then by (\ref{sht_6}) the characteristic surfaces of the
principal symbol are (notice that $\mathfrak{N}_{~\alpha }^{Ab}$ is real) 
\begin{eqnarray}
0 &=&\sqrt{\det \left( G_{2}^{\alpha \gamma }\mathfrak{N}_{~\gamma }^{Aa}%
\bar{n}_{a}\left( \lambda \right) G_{1AB}\mathfrak{N}_{~\beta
}^{Bb}n_{b}\left( \lambda \right) \right) }  \notag \\
&=&\left( \left\vert n_{0}\right\vert
^{2}+n_{1}^{2}+n_{2}^{2}+n_{3}^{2}\right) ^{2}\left\vert \left(
-n_{0}^{2}+n_{3}^{2}\right) \right\vert \allowbreak \left\vert
n_{a}g^{ab}n_{b}\right\vert l_{1}^{2}l_{2}^{2}.  \label{FF_m_1}
\end{eqnarray}

It means that the characteristic surfaces are given in terms of two
symmetric tensors, the background metric, and $g_{1}^{ab}=diag\left(
-1,0,0,1\right) $ i.e.%
\begin{equation}
\begin{array}{ccccc}
0=n_{a}g^{ab}n_{b} &  & \text{and} &  & 0=n_{a}g_{1}^{ab}n_{b}%
\end{array}
\label{ff_m_2}
\end{equation}
The first one corresponds to the electromagnetic waves and the second one to
the Alfven waves. Because the characteristic eigenvalues are real, thus the
system is hyperbolic.

Note that the introduction of two unnatural scalar products lead us to a
preferred Euclidean metric $g_{2}^{ab}n_{a}n_{b}=\left\vert n_{0}\right\vert
^{2}+n_{1}^{2}+n_{2}^{2}+n_{3}^{2}.$

2) We shall check that condition 2 in Theorem \ref{sht_Theo_1} fails: For
this system, it is possible to calculate the singular values. We only show
the relevant one (with $n_{0}$ real)

\begin{eqnarray*}
\sigma \left[ \mathfrak{N}_{~\beta }^{Bb}n_{b}\right] &:&=\frac{1}{\sqrt{2}}%
\left\vert \sqrt{N+2\left( n_{0}^{2}-n_{3}^{2}\right) l_{2}^{2}}-\sqrt{%
N-2\left( n_{0}^{2}-n_{3}^{2}\right) l_{2}^{2}}\right\vert \\
N &:&=n_{1}^{2}+n_{2}^{2}+\left( n_{0}^{2}+n_{3}^{2}\right) \left(
1+l_{2}^{4}\right)
\end{eqnarray*}

Notice that it vanishes when $n_{0}^{2}-n_{3}^{2}=0.$

Consider now the line $n_{a}\left( \lambda \right) =-\lambda \omega
_{a}+\beta _{a}\in S_{\omega _{a}}$ with $\lambda $ real, $\beta
_{a}=n_{1}e_{1a}$ and the characteristic eigenvalue $\lambda =0$, i.e.$%
\left. n_{a}\left( \lambda \right) g_{1}^{ab}n_{b}\left( \lambda \right)
\right\vert _{\lambda =0}=0$. Perturbing this singular value in a
neighborhood of this point $n_{\varepsilon ,\theta }\left( \lambda =0\right)
_{a}=\left. -\varepsilon e^{i\theta }\omega _{a}-\lambda \omega _{a}+\beta
_{a}\right\vert _{\substack{ \lambda =0  \\ \theta =0}}$, we obtain 
\begin{equation*}
\sigma _{1}\left( \varepsilon \right) \approx \frac{1}{\sqrt{2}}\varepsilon
^{2}\left( \frac{\left( 1+3l_{2}^{4}\right) }{\sqrt{n_{1}^{2}+\varepsilon
^{2}\left( 1+3l_{2}^{4}\right) }}-\frac{\left( 1-l_{2}^{4}\right) }{\sqrt{%
n_{1}^{2}+\varepsilon ^{2}\left( 1-l_{2}^{4}\right) }}\right)
\end{equation*}

It is order $O\left( \left\vert \varepsilon \right\vert ^{2}\right) $ and by
Theorem \ref{sht_Theo_1} there does not exist any hyperbolizer and the
system is weakly hyperbolic.

In general, explicit calculations of the singular values can not be done.
Because of that, we shall show how to reach the same conclusion using
theorem \ref{Conjetura_final}.

Consider the line $n_{a}\left( \lambda \right) $ as before, then we get the
following principal symbol%
\begin{eqnarray*}
\mathfrak{N}_{~\alpha }^{Aa}n_{a}\left( \lambda \right) &=&-\lambda 
\mathfrak{N}_{~\alpha }^{Aa}\omega _{a}+\mathfrak{N}_{~\alpha }^{Aa}\beta
_{a} \\
&=&-\lambda \left( 
\begin{array}{cc}
0 & l_{1}^{2}e_{0a} \\ 
-l_{2}^{2}e_{0a} & 0 \\ 
e_{0}^{[b}\delta _{a}^{c]} & 0 \\ 
0 & e_{0}^{[b}\delta _{a}^{c]}%
\end{array}%
\right) +\left( 
\begin{array}{cc}
0 & 0 \\ 
l_{2}^{2}n_{1}e_{1a} & -l_{1}l_{2}n_{1}e_{2a} \\ 
n_{1}e_{1}^{[b}\delta _{a}^{c]} & 0 \\ 
0 & n_{1}e_{1}^{[b}\delta _{a}^{c]}%
\end{array}%
\right)
\end{eqnarray*}

To define $L_{~i}^{j}$ as (\ref{conj_L}), we need to calculate left and
right kernel basis of $\mathfrak{N}_{~\alpha }^{Aa}n_{a}\left( \lambda
=0\right) $. They are 
\begin{equation*}
\left( 
\begin{array}{c}
\delta W \\ 
\delta X \\ 
\delta Y_{bc} \\ 
\delta Z_{bc}%
\end{array}%
\right) =\left\langle 
\begin{array}{c}
\left( 
\begin{array}{c}
1 \\ 
0 \\ 
0 \\ 
0%
\end{array}%
\right) ,\left( 
\begin{array}{c}
0 \\ 
0 \\ 
e_{0[b}e_{2a]} \\ 
0%
\end{array}%
\right) ,\left( 
\begin{array}{c}
0 \\ 
0 \\ 
e_{0[b}e_{3a]} \\ 
0%
\end{array}%
\right) ,\left( 
\begin{array}{c}
0 \\ 
0 \\ 
e_{2[b}e_{3a]} \\ 
0%
\end{array}%
\right)  \\ 
,\left( 
\begin{array}{c}
0 \\ 
0 \\ 
0 \\ 
e_{0[b}e_{2a]}%
\end{array}%
\right) ,\left( 
\begin{array}{c}
0 \\ 
0 \\ 
0 \\ 
e_{0[b}e_{3a]}%
\end{array}%
\right) ,\left( 
\begin{array}{c}
0 \\ 
0 \\ 
0 \\ 
e_{2[b}e_{3a]}%
\end{array}%
\right) 
\end{array}%
\right\rangle 
\end{equation*}%
\begin{equation*}
\left( 
\begin{array}{c}
\delta l_{1}^{a} \\ 
\delta l_{2}^{a}%
\end{array}%
\right) =\left\langle \left( 
\begin{array}{c}
0 \\ 
e_{1}^{a}%
\end{array}%
\right) \right\rangle 
\end{equation*}

We conclude that%
\begin{equation}
L_{~i}^{j}=\left( 
\begin{array}{c}
0 \\ 
0 \\ 
0 \\ 
0 \\ 
0 \\ 
0 \\ 
0%
\end{array}%
\right)  \label{L_1}
\end{equation}

which trivially vanishes and then it has right kernel. Thus, as we discussed
before, there is a singular value that goes to zero at least quadratically
in the perturbation and the system can not be strongly hyperbolic.

If we take $\omega _{a}$ outside the light cone, then there will be complex
characteristic eigenvalues (so the system would not be hyperbolic along
those lines), so those cases are trivially not strongly hyperbolic.

\subsection{Charged fluids with finite conductivity. \label{Example_fluids}}

In this subsection, we present the charged fluid with finite conductivity in a
first order in derivative formulation, in which the relevant block of the
principal part has no constraints. We shall prove that the system is weakly
hyperbolic while it has finite conductivity and, of course, strongly
hyperbolic with vanishing conductivity. This result is in concordance with 
\cite{choquet2009general} chapter IX.

The system is 
\begin{eqnarray*}
u^{m}\nabla _{m}n+n\nabla _{m}u^{m} &=&0 \\
u^{a}\nabla _{a}\rho +\left( \rho +p\right) \nabla _{a}u^{a}
&=&u_{b}J^{a}F_{~a}^{b} \\
\left( \rho +p\right) u^{a}\nabla _{a}u^{b}+D^{b}p
&=&-h_{c}^{b}J^{a}F_{~a}^{c} \\
u^{m}\nabla _{m}q+q\nabla _{m}u^{m}+\sigma F_{a}^{~m}\nabla _{m}u^{a}
&=&\sigma u^{a}J_{a} \\
\nabla _{a}F^{ab} &=&J^{b} \\
\nabla _{a}F^{\ast ab} &=&0 \\
J^{a} &=&qu^{a}+\sigma u_{b}F^{ba}
\end{eqnarray*}

with background metric $g_{ab}$, $h_{~c}^{b}:=\left( \delta
_{c}^{b}+u^{b}u_{c}\right) ,$ $\ u^{a}u_{a}=-1,$ $\ D^{b}:=h^{bc}\nabla _{a}$
and $p=p\left( n,\rho \right) .$ Here $\rho $ is the proper total energy
density, $n$ the proper mass density, $u^{a}$ the four-velocity, $q$ the
proper charge density, $p$ the pressure of the fluid and $\sigma $ the
conductivity. For examples of this type of systems see \cite%
{yamada2010magnetic}, \cite{palenzuela2013modelling}, \cite%
{palenzuela2009beyond}.

The variables are $\left( n,\rho ,u^{a},q,F^{ab}\right) .$ As before, taking
the linearized version at a given point and background solution of these
equations, the principal symbol is given by%
\begin{eqnarray*}
\left( \mathfrak{N}_{fluid}\right) _{~\alpha }^{Aa}\left( x,\phi \right)
n_{a}\delta \phi ^{\alpha } &=&\left( 
\begin{array}{cccc}
u.n & 0 & n\text{ }n_{b} & 0 \\ 
0 & u.n & \left( p+\rho \right) n_{b} & 0 \\ 
p_{n}h^{am}n_{m} & p_{\rho }h^{am}n_{m} & \left( \rho +p\right) \delta
_{~b}^{a}\left( u.n\right) & 0 \\ 
0 & 0 & \left( q\delta _{b}^{m}+\sigma F_{b}^{~m}\right) n_{m} & u.n%
\end{array}%
\right) \left( 
\begin{array}{c}
\delta n \\ 
\delta \rho \\ 
\delta u^{b} \\ 
\delta q \\ 
\delta F^{ab}%
\end{array}%
\right) =0 \\
\left( \mathfrak{N}_{Electro}\right) _{~\alpha }^{Aa}\left( x,\phi \right)
n_{a}\delta \phi ^{\alpha } &=&\left( 
\begin{array}{c}
n_{a} \\ 
n_{c}\varepsilon _{~dab}^{c}%
\end{array}%
\right) \delta F^{ab}
\end{eqnarray*}

with $u_{a}\delta u^{a}=0.$ Notice that the fluid-current part decouples of
the electrodynamics part. We shall only study this fluid-current part
because there is where the lack of strong hyperbolicity appears. This part
of the system has no constraints.

The solution space $\delta \phi ^{\alpha }=\left( 
\begin{array}{c}
\delta n \\ 
\delta \rho \\ 
\delta u^{b} \\ 
\delta q%
\end{array}%
\right) $ is $6$-dimensional and the equation space

$\delta X_{A}=\left( 
\begin{array}{cccc}
\delta W & \delta X & \delta Y_{a} & \delta Z%
\end{array}%
\right) $, with $\delta Y_{a}u^{a}=0$, is $6$-dimensional too. \ 

The characteristic surfaces of the fluids-current part are 
\begin{equation}
\det \left( \mathfrak{N}_{fluid~\alpha }^{Aa}n_{a}\right) =-\left( \rho
+p\right) ^{4}\left( n_{a}u^{a}\right) ^{4}g_{1}^{ab}n_{a}n_{b}=0
\label{det_fluid}
\end{equation}

This means that 
\begin{equation*}
\begin{array}{ccccc}
\left( n_{a}u^{a}\right) =0 &  & \text{and} &  & g_{1}^{ab}n_{a}n_{b}=0%
\end{array}%
\end{equation*}%
with $g_{1}^{ab}:=\left( \frac{n}{\left( \rho +p\right) }p_{n}+p_{\rho
}\right) h^{ab}-u^{a}u^{b}$ (It is a Lorentzian metric if $\frac{n}{\left(
\rho +p\right) }p_{n}+p_{\rho }>0$). In addition the characteristic surfaces
of the electrodynamics part are 
\begin{equation*}
g^{ab}n_{a}n_{b}=0
\end{equation*}

The $n_{a}u^{a}=0$ correspond to the material waves, $g_{1}^{ab}n_{a}n_{b}=0$
to the acoustic waves and $g^{ab}n_{a}n_{b}=0$ to the electromagnetic waves.

1) We shall check condition 1 in theorem \ref{Teo_no_jordan}: Consider now
the line $n_{a}\left( \lambda \right) =-\lambda \omega _{a}+\beta _{a}\in
S_{\omega _{a}}$ with $\omega _{a}=u_{a}$ and $\beta _{a}$ spacelike and
such that $\beta _{a}u^{a}=0.$ We notice from (\ref{det_fluid}) that $%
\mathfrak{N}_{fluid~\alpha }^{Aa}\omega _{a}$ has no right kernel if $\left(
\rho +p\right) \neq 0$ and the system is hyperbolic for this $\omega _{a}$
if $\frac{n}{\left( \rho +p\right) }p_{n}+p_{\rho }\geq 0$. It means that
the velocity of the acoustic wave $v:=\sqrt{\frac{n}{\left( \rho +p\right) }%
p_{n}+p_{\rho }}$ is real.

2) Condition 2 in theorem \ref{Teo_no_jordan} fails: This line has the
characteristic eigenvalue $\lambda =0$, since $\left. u^{a}n_{a}\left(
\lambda \right) \right\vert _{\lambda =0}=0$. \ We choose an orthonormal
frame $\left\{ e_{ia}\text{ \ }i=0,1,2,3\right\} $ such that $e_{0a}=u_{a}$, 
$\ e_{1a}=\frac{1}{\sqrt{\beta ^{a}\beta _{a}}}\beta _{a}$ with $e_{2a}$ and 
$e_{3a}$ space-like.$\ $In this frame the background metric looks like $%
g_{ab}=diag\left( -1,1,1,1\right) .$

The principal symbol along this line is

\begin{eqnarray*}
\mathfrak{N}_{~\alpha }^{Aa}n_{a}\left( \lambda \right) &=&-\lambda 
\mathfrak{N}_{~\alpha }^{Aa}\omega _{a}+\mathfrak{N}_{~\alpha }^{Aa}\beta
_{a} \\
&=&-\lambda \left( 
\begin{array}{cccc}
-1 & 0 & nu_{b} & 0 \\ 
0 & -1 & \left( p+\rho \right) u_{b} & 0 \\ 
0 & 0 & -\left( \rho +p\right) g_{~~b}^{a} & 0 \\ 
0 & 0 & qu_{b}+\sigma F_{b}^{~m}u_{m} & -1%
\end{array}%
\right) +\sqrt{\beta .\beta }\left( 
\begin{array}{cccc}
0 & 0 & ne_{1b} & 0 \\ 
0 & 0 & \left( p+\rho \right) e_{1b} & 0 \\ 
p_{n}e_{1}^{a} & p_{\rho }e_{1}^{a} & 0 & 0 \\ 
0 & 0 & qe_{1b}+\sigma F_{b}^{~m}e_{1m} & 0%
\end{array}%
\right)
\end{eqnarray*}

In order to find $L_{~i}^{j}$ the basis of the left and right kernel of $%
\mathfrak{N}_{~\alpha }^{Aa}n_{a}\left( \lambda =0\right) $ are%
\begin{equation*}
\left( 
\begin{array}{c}
\delta W \\ 
\delta X \\ 
\delta Y_{a} \\ 
\delta Z%
\end{array}%
\right) =\left\langle \left( 
\begin{array}{c}
0 \\ 
0 \\ 
e_{2a} \\ 
0%
\end{array}%
\right) ,\left( 
\begin{array}{c}
0 \\ 
0 \\ 
e_{3a} \\ 
0%
\end{array}%
\right) ,\left( 
\begin{array}{c}
-\left( p+\rho \right) \\ 
n \\ 
0 \\ 
0%
\end{array}%
\right) \right\rangle
\end{equation*}%
\begin{equation*}
\left( 
\begin{array}{c}
\delta n \\ 
\delta \rho \\ 
\delta u^{b} \\ 
\delta q%
\end{array}%
\right) =\left\langle \left( 
\begin{array}{c}
0 \\ 
0 \\ 
0 \\ 
1%
\end{array}%
\right) ,\left( 
\begin{array}{c}
-p_{\rho } \\ 
p_{n} \\ 
0 \\ 
0%
\end{array}%
\right) ,\left( 
\begin{array}{c}
0 \\ 
0 \\ 
\delta u^{b}\text{ } \\ 
0%
\end{array}%
\right) \right\rangle
\end{equation*}

with $e_{1b}\delta u^{b}=0$ and $\ \delta u^{b}F_{b}^{~m}e_{1m}=0$. \ Thus
following equation (\ref{conj_L}) 
\begin{equation*}
L_{~i}^{j}=\left( 
\begin{array}{ccc}
0 & 0 & -\left( \rho +p\right) \delta u^{a}e_{2a} \\ 
0 & 0 & -\left( \rho +p\right) \delta u^{a}e_{3a} \\ 
0 & -\left( rp_{\rho }+np_{n}\right)  & 0%
\end{array}%
\right) 
\end{equation*}

Clearly $L_{~i}^{j}\left( 
\begin{array}{c}
1 \\ 
0 \\ 
0%
\end{array}%
\right) =0$. \ Therefore the system is weakly hyperbolic and there is no
hyperbolizer.

Notice that we chose a particular $\omega _{a}$ $=u_{a}$. It is easy to show
that condition 2\ still fails for any timelike $\omega _{a}$ in both metrics 
$g^{ab}\omega _{a}\omega _{a}<0$ and $g_{1}^{ab}\omega _{a}\omega _{a}<0$.
In addition, when choosing $\omega _{a}$ outside of these cones, complex
characteristic eigenvalues appear. Thus, in both cases no hyperbolizer
exists.

\subsubsection{Vanish conductivity $\protect\sigma =0.$}

Finally, we notice that if the conductivity goes to zero $\sigma =0$ the
kernels change and the system becomes strongly hyperbolic.

Consider the above line $n_{a}\left( \lambda \right) $ in $\lambda =0$, thus
we shall prove that the new $L_{~i}^{j}$ has no right kernel. The new left
and right kernel basis \ are

\begin{equation*}
\left( 
\begin{array}{c}
\delta W \\ 
\delta X \\ 
\delta Y_{a} \\ 
\delta Z%
\end{array}%
\right) =\left\langle \left( 
\begin{array}{c}
0 \\ 
0 \\ 
e_{2a} \\ 
0%
\end{array}%
\right) ,\left( 
\begin{array}{c}
0 \\ 
0 \\ 
e_{3a} \\ 
0%
\end{array}%
\right) ,\left( 
\begin{array}{c}
-\left( p+\rho \right) \\ 
n \\ 
0 \\ 
0%
\end{array}%
\right) ,\left( 
\begin{array}{c}
0 \\ 
q \\ 
0 \\ 
-\left( p+\rho \right)%
\end{array}%
\right) \right\rangle
\end{equation*}%
\begin{equation*}
\text{\ \ \ \ }\left( 
\begin{array}{c}
\delta n \\ 
\delta \rho \\ 
\delta u^{b} \\ 
\delta q%
\end{array}%
\right) =\left\langle \left( 
\begin{array}{c}
0 \\ 
0 \\ 
0 \\ 
1%
\end{array}%
\right) ,\left( 
\begin{array}{c}
-p_{\rho } \\ 
p_{n} \\ 
0 \\ 
0%
\end{array}%
\right) ,\left( 
\begin{array}{c}
0 \\ 
0 \\ 
e_{2}^{a}\text{ } \\ 
0%
\end{array}%
\right) ,\left( 
\begin{array}{c}
0 \\ 
0 \\ 
e_{3}^{a}\text{ } \\ 
0%
\end{array}%
\right) \right\rangle
\end{equation*}

Thus%
\begin{equation*}
L_{~i}^{j}=\left( 
\begin{array}{cccc}
0 & 0 & -\left( \rho +p\right) & 0 \\ 
0 & 0 & 0 & -\left( \rho +p\right) \\ 
0 & -\left( \left( p+\rho \right) p_{\rho }+np_{n}\right) & 0\text{ } & 0%
\text{ } \\ 
-\left( p+\rho \right) & -p_{n}q & 0 & 0%
\end{array}%
\right)
\end{equation*}

This operator has no kernel if the determinant is different from zero 
\begin{equation*}
\det L_{~i}^{j}=-\left( p+\rho \right) ^{4}\left( \frac{n}{\left( \rho
+p\right) }p_{n}+p_{\rho }\right) \neq 0
\end{equation*}

Therefore $p+\rho \neq 0$ and $\frac{n}{\left( \rho +p\right) }p_{n}+p_{\rho
}\neq 0.$ As we explained, the first condition is necessary in order for $%
\mathfrak{N}_{fluid~\alpha }^{Aa}\omega _{a}$ not to have right kernel and
the second condition limits the possibility of the velocity of the acoustic
waves to vanish. We conclude that the system with $\sigma =0$ is strongly
hyperbolic.

\bigskip

\section{CONCLUSIONS.}

In this article we studied the covariant theory of strong hyperbolicity for
systems with constraints and found a necessary condition for the systems to
admit a hyperbolizer. If this condition, which is easy to check, is not
satisfied, then there is no subset of evolution equations of the strongly
hyperbolic type in the usual sense.

To find this condition we introduce the singular value decomposition of one
parameter families (pencils) of principal symbols and study perturbations
around the points where they have kernel. We proved that if the
perturbations of a singular value, which vanishes at that point, are order $2
$ or larger then the system is not strongly hyperbolic, theorem \ref%
{sht_Theo_1}.

For systems with constraints, the rectangular case, is only a necessary
condition, but in the case without constraints, namely square case, this
condition becomes also a sufficient condition too, theorem \ref%
{Teo_no_jordan}. In this case the condition is equivalent to the ones in
Kreiss's Matrix Theorem.

As an extra result, we showed that a perturbed Matrix has an $l$-Jordan
Block if and only if it has a singular value of order $O\left( \varepsilon
^{l}\right) .$

Although the SVD depends on the scalar products used to define the adjoint
operators, we found that the asymptotic orders of the singular values are
independent of them.

When the systems have constraints, their principal symbols are rectangular
operators and there is not a simple way to find their characteristic
surfaces. We proposed a way to calculate them, by connecting the kernel of
an operator with the vanishing of any of their singular values (see equation
(\ref{sht_6})).

We applied these theorems to some examples of physical interest:

A simple matrix example of $2\times 2$ in which we study its Jordan form
using the perturbed SVD.

The second example is Force Free electrodynamic in the Euler potentials
form, that when written in first order form has constraints. Using our
result we checked that there is no hyperbolizer, being the system only
weakly hyperbolic. We did this in two alternative ways. First by computing
the singular values and showing that at some point one of them is order $%
O\left( \varepsilon ^{2}\right) $. Second by computing the first order
leading term\ using right and left kernels as in theorem \ref%
{Conjetura_final}. In general, it is not an easy task to calculate the
singular values, therefore, the second way simplifies the study of strong
hyperbolicity.

The last example is a charged fluid with finite conductivity. For this case,
it is enough to consider the fluid-current part that decouples (at the level
of the principal symbol) from the electromagnetic part. We showed how the
introduction of finite conductivity, hampers the possibility of a
hyperbolizer.


\section*{ACKNOWLEDGMENTS}

I would like to specially thank Oscar Reula for the uncountable discussions,
his indispensable help through this work, and for his predisposition to read
and correct the article. I further would like to thank Federico Carrasco,
Marcelo Rubio, Barbara Sbarato and Miguel Megevand for the meetings in which
we have been able to exchange a set of interesting ideas. This work was
supported by SeCyT-UNC, CONICET and FONCYT.

\appendix

\section{Singular value decomposition. \label{appendix}}

In this appendix, the Singular Value Decomposition (SVD) is defined for
linear operators that map between two vector spaces of finite dimension. For
introduction on topic see \cite{golub2012matrix}, \cite%
{stewart1973introduction}, \cite{stewart1990matrix}. \cite%
{watkins2004fundamentals}. One of the most significant properties of the SVD
is that it allows us to characterize the image and the kernel of the
operator through real quantities called singular values. They, as we showed
in section \ref{singular_value_}, provide data about the Jordan form of
square matrices.

One problem about singular values decomposition, is that, it is necessary to
introduce extra structure to the problem, namely, scalar products. When they
are used in vector spaces over a manifold, they might introduce
non-covariant expressions. These scalar products are two positive definite,
tensorial Hermitian forms\footnote{%
Consider the $%
\mathbb{C}
-$vectorial space $V$ of finite dimension. A Hermitian form on $V$ is a map $%
G:V\times V\rightarrow 
\mathbb{C}
$ such that $G\left( av,bu\right) =\bar{a}bG\left( v,u\right) =\bar{a}%
b\left( \bar{v}^{a}G_{ab}u^{b}\right) $ and $G_{ab}=\bar{G}_{ab}$ (the bar
means conjugation)$.$ In addition there exists a symmetric real bilinear
form $g_{1ab}$ and an antisymmetric bilinear form $g_{2ab}$ such that $%
G_{ab}=g_{1ab}+ig_{2ab}$. When the hermite form is positive definite ($%
g_{1ab}$ is positive define), $G$ defines a complex inner dot product of $V$.%
}, in the input and output spaces of the operator, respectively.

In this appendix, we use the notation of section \ref{singular_value_} with $%
K_{~\alpha }^{A}:X\rightarrow E$ and assume that $e=\dim E\geq \dim X=u.$

Consider the two positive definite Hermitian forms $G_{1AB}$ and $G_{2\alpha
\beta }$ in the spaces $E$ and $X$ respectively. This allows us to define 
\textbf{the adjoint operator} 
\begin{eqnarray*}
K^{\ast } &=&G_{2}^{-1}\circ \bar{K}^{\prime }\circ G_{1}:E\rightarrow X \\
\left( K^{\ast }\right) _{~C}^{\alpha } &=&\left( G_{2}^{-1}\circ \bar{K}%
^{\prime }\circ G_{1}\right) _{~C}^{\alpha }=G_{2}^{\alpha \beta }\bar{K}%
_{~\beta }^{B}G_{1BC}
\end{eqnarray*}

Where $G_{1}^{AB}G_{1BC}=\delta _{1C}^{A}$ and $G_{2}^{\alpha \gamma
}G_{2\gamma \nu }=\delta _{~\nu }^{\alpha }$ are the identity operators in $%
E $ and $X$ respectively, and $\bar{K}^{\prime }$ is the dual complex
operator of $K.$

With this operator we can define, 
\begin{eqnarray*}
K^{\ast }\circ K &=&G_{2}^{-1}\circ K^{\prime }\circ G_{1}\circ
K:X\rightarrow X \\
\phi ^{\alpha } &\rightarrow &G_{2}^{\alpha \beta }\bar{K}_{~\beta
}^{B}G_{1BC}K_{~\gamma }^{C}\phi ^{\gamma } \\
K\circ K^{\ast } &=&K\circ G_{2}^{-1}\circ K^{\prime }\circ
G_{1}:E\rightarrow E \\
l^{A} &\rightarrow &K_{~\alpha }^{A}G_{2}^{\alpha \beta }\bar{K}_{~\beta
}^{B}G_{1BC}l^{C}
\end{eqnarray*}

Since $G_{2}\circ K^{\ast }\circ K$ and $G_{1}\circ K\circ K^{\ast }$ are
semi-positive define Hermitian forms, $K^{\ast }\circ K$ and $K\circ K^{\ast
}$ are diagonalizable \ with real and semi-positive eigenvalues. Also, the
square roots of these eigenvalues are the singular values of $K$ and $%
K^{\ast }$.

With these definitions, we assert the singular value decomposition in the
form of a Theorem. From now on Latin indices $i,j,k$ go from $1$ to $e$ and
primes Latin indices $i^{\prime },j^{\prime },k^{\prime }$ from $\ 1$ to $u$%
, unless explicitly stated. These indices indicate the different
eigenvectors.

\begin{theorem}
\label{SVD_theorem}Consider\ $K,K^{\ast },G_{1}$and $G_{2}$ as previously
defined and $e\geq u$. Suppose that \ \thinspace $rank\left( K\right) =r$
and $\dim \left( \ker \_right\left( K\right) \right) =u-r$ , then $K$ can be
decomposed as 
\begin{equation*}
K_{~\alpha }^{A}=U_{~i}^{A}\Sigma _{~j^{\prime }}^{i}\left( V^{-1}\right)
_{~\alpha }^{j^{\prime }}
\end{equation*}

where $\Sigma _{~j^{\prime }}^{i}=\left( 
\begin{array}{cc}
\Sigma _{+} & 0 \\ 
0 & 0 \\ 
0 & 0%
\end{array}%
\right) _{~j^{\prime }}^{i}$ of size $e\times u,$ with $\ \Sigma
_{+}=diag\left( \sigma _{1},...\sigma _{r}\right) $, $\sigma _{1}\geq
...\geq \sigma _{l}>0$ $\ l=1,...,r$ real, and $\sigma _{r+1}=...=\sigma
_{u}=0$. The $\sigma _{l}$ are called singular values of $K$ and they are
the square root of the eigenvalues of $K^{\ast }\circ K.$

In addition, the columns of $U_{~i}^{A}$ and $V_{~i^{\prime }}^{\gamma }$
are eigenbasis of $K\circ K^{\ast }$ and $K^{\ast }\circ K$ \ respectively,
such that they are orthogonal%
\begin{eqnarray}
\bar{U}_{~i}^{C}G_{1CD}U_{~j}^{D} &=&\delta _{1ij}:=diag\left(
1,,...,1,1,...,1\right)  \label{a_lem3_a} \\
\bar{V}_{~i^{\prime }}^{\gamma }G_{2\gamma \eta }V_{~j^{\prime }}^{\eta }
&=&\delta _{2i^{\prime }j^{\prime }}:=diag\left( 1,...,1\right)
\label{a_lem3_b}
\end{eqnarray}

with $C,D,i,j=1,..,e$ \ and $\alpha ,\beta \,i^{\prime },j^{\prime
}=1,...,u. $

We are going to discriminate $V_{~i^{\prime }}^{\alpha }=\left(
V_{2},V_{1}\right) _{~i^{\prime }}^{\alpha }$ and $U_{~i}^{A}=\left(
U_{2},U_{1},U_{3}\right) _{~i}^{A},$ where $V_{2}$ are the first $r$ columns
and $V_{1}$ the $u-r$ left\ of $V;$ $U_{2}$ are the first $r$ columns, $%
U_{1} $ the following $u-r$ and $U_{3}$ the remaining $e-u$ of $U.$
\end{theorem}

Recalling that the eigenbasis that are chosen in the Jordan decomposition
are not unique we realize that neither they are the orthogonal factors in
the SVD. For fixed $G_{1,2}$ we can select different\ orthogonal basis of
the eigenspaces, associated to some singular value, and obtain different $%
U,V $. Nevertheless the singular values remain invariant as long as $G_{1,2}$
remain fixed.

This decomposition allows us to control right and left kernels and images of
any linear operator, as we show in the next Corollary.

\begin{corollary}
\label{Cor_ker}Consider $V=\left( V_{2},V_{1}\right) _{~i^{\prime }}^{\alpha
}$ and $U=\left( U_{2},U_{1},U_{3}\right) _{~i}^{A}$ as in the previous
theorem, then the orthogonal conditions (\ref{a_lem3_a}) and (\ref{a_lem3_b}%
) are%
\begin{equation}
\left( \bar{V}_{i}\right) _{~s}^{\gamma }G_{2\gamma \eta }\left(
V_{j}\right) _{~r}^{\eta }=0\text{ with }i,j=1,2\text{ and }i\neq j
\label{cor_e_1}
\end{equation}%
\begin{equation}
\left( \bar{U}_{i}\right) _{~s}^{C}G_{1CD}\left( U_{j}\right) _{r}^{D}=0%
\text{ with }i,j=1,2,3\text{ and }i\neq j  \label{cor_e_2}
\end{equation}

In addition 
\begin{eqnarray*}
\dim \left( right\_\ker \left( K\right) \right) &=&u-r \\
\dim \left( left\_\ker \left( K\right) \right) &=&e-r \\
\dim \left( rank\_Columns\left( K\right) \right) &=&\dim \left(
rank\_Rows\left( K\right) \right) =r
\end{eqnarray*}

And the explicit right and left kernels of $K$ are%
\begin{eqnarray}
K_{~\beta }^{A}V_{1s}^{\beta } &=&0\text{ \ \ with }s=1,..,u-r
\label{Cor_ker_1a} \\
\left( \delta _{1}^{lm}\bar{U}_{1m}^{C}G_{1CB}\right) K_{~\beta }^{B} &=&0%
\text{ \ \ with }m,l=1,..,u-r  \label{Cor_ker_2a} \\
\left( \delta _{1}^{lm}\bar{U}_{3m}^{C}G_{1CB}\right) K_{~\beta }^{B} &=&0%
\text{ \ \ with }m,l=1,..,e-u  \label{Cor_ker_3a}
\end{eqnarray}

where $\delta _{1}^{lm}$ is the inverse of $\delta _{1lm}$
\end{corollary}

SVD is similar to Jordan decomposition for square operators. In particular,
they coincide when the operators are diagonalizable with real and
semipositive eigenvalues and particular $G_{1,2}$ are used.

Suppose $A_{~\beta }^{\alpha }:X\rightarrow X$ \ can be decomposed as\ $%
A_{~\beta }^{\alpha }=P_{~i^{\prime }}^{\alpha }\Lambda _{~j^{\prime
}}^{i^{\prime }}\left( P^{-1}\right) _{~\beta }^{j^{\prime }}$ with $\Lambda
_{~j^{\prime }}^{i^{\prime }}=diag\left( \lambda _{1},...,\lambda
_{u}\right) $ and $\lambda _{i}$ real and semipositive. If we choose $%
G_{1\alpha \beta }=G_{2\alpha \beta }=\overline{\left( P^{-1}\right)
_{~\alpha }^{i^{\prime }}}\delta _{i^{\prime }j^{\prime }}\left(
P^{-1}\right) _{~\beta }^{j^{\prime }}$ then 
\begin{eqnarray*}
\left( A^{\ast }\circ A\right) _{~\beta }^{\alpha } &=&G_{2}^{\alpha \gamma }%
\overline{A_{~\gamma }^{\eta }}G_{1\eta \varphi }A_{~\beta }^{\varphi } \\
&=&\left( P_{~i_{1}^{\prime }}^{\alpha }\delta ^{i_{1}^{\prime
}j_{1}^{\prime }}\bar{P}_{~j_{1}^{\prime }}^{\gamma }\right) \left( \bar{P}%
_{~i_{2}^{\prime }}^{\eta }~\bar{\Lambda}_{~j_{2}^{\prime }}^{i_{2}^{\prime
}}~\left( \bar{P}^{-1}\right) _{~\gamma }^{j_{2}^{\prime }}\right) \left(
\left( \bar{P}^{-1}\right) _{~\eta }^{i_{3}^{\prime }}\delta _{i_{3}^{\prime
}j_{3}^{\prime }}\left( P^{-1}\right) _{~\varphi }^{j_{3}^{\prime }}\right)
\left( P_{~i^{\prime }}^{\varphi }\Lambda _{~j^{\prime }}^{i^{\prime
}}\left( P^{-1}\right) _{~\beta }^{j^{\prime }}\right) \\
&=&P_{~i^{\prime }}^{\alpha }\left( \delta ^{i^{\prime }j^{\prime }}%
\overline{\Lambda _{~j^{\prime }}^{k^{\prime }}}\delta _{k^{\prime
}i_{1}^{\prime }}\Lambda _{~j_{2}^{\prime }}^{i_{1}^{\prime }}\right) \left(
P^{-1}\right) _{~\beta }^{j_{2}^{\prime }} \\
&=&P_{~i^{\prime }}^{\alpha }\left( \Lambda ^{2}\right) _{~j^{\prime
}}^{i^{\prime }}\left( P^{-1}\right) _{~\beta }^{j^{\prime }}
\end{eqnarray*}

Therefore the eigenvalues of $\left( A^{\ast }\circ A\right) $ are $\lambda
_{i}^{2},$ thus the singular values of $A$ are $\lambda _{i}$. In addition $%
U_{~i^{\prime }}^{\alpha }=V_{~i^{\prime }}^{\alpha }=P_{~i^{\prime
}}^{\alpha }$ and the orthogonal condition is $\overline{P_{~i^{\prime
}}^{\gamma }}G_{1,2\gamma \eta }P_{~j^{\prime }}^{\eta }=\delta _{i^{\prime
}j^{\prime }}$.\footnote{%
The standard result in textbook is when $G_{1\alpha \beta }=G_{2\alpha \beta
}=\delta _{\alpha \beta }=\overline{\left( P^{-1}\right) _{~\alpha
}^{i^{\prime }}}\delta _{i^{\prime }j^{\prime }}\left( P^{-1}\right)
_{~\beta }^{j^{\prime }}$ it means that $\left( P\right) _{~i^{\prime
}}^{\alpha }$ is orthogonal, and the matrix $A_{~\beta }^{\alpha }$ is
"symmetric".}

Notice that, in the deduction, we used $\left( \delta ^{i^{\prime }j^{\prime
}}\overline{\Lambda _{~j^{\prime }}^{k^{\prime }}}\delta _{k^{\prime
}i_{1}^{\prime }}\Lambda _{~j_{2}^{\prime }}^{i_{1}^{\prime }}\right)
=\left( \Lambda ^{2}\right) _{~j^{\prime }}^{i^{\prime }}$ because $%
A_{~\beta }^{\alpha }$ is diagonalizable, with real eigenvalues. But if $%
\Lambda _{~j^{\prime }}^{i^{\prime }}$\ is a non trivial Jordan form, then
the singular values are the square roots of the eigenvalues of $\delta
^{i^{\prime }j^{\prime }}\overline{\Lambda _{~j^{\prime }}^{k^{\prime }}}%
\delta _{k^{\prime }i_{1}^{\prime }}\Lambda _{~j_{2}^{\prime
}}^{i_{1}^{\prime }}$, the explicit calculation becomes hard even for simple
examples. In chapter \ref{chap_examples} we present an analysis of the $%
2\times 2$ matrix case.

\bigskip

\section{Invariant orders of singular values.   \label{appendix_2}}

In this article, we study perturbed singular values in terms of some
parameter $\varepsilon .$ We use the orders in this parameter, to decide
when a system is strongly hyperbolic or not. But as we showed in appendix %
\ref{appendix}, the singular values depend on two Hermitian forms, therefore
we need to show that these orders remain invariant when we select different
Hermitian forms. Thus we shall prove it in Lemma \ref{Lemma_metricas}. Also,
in theorem \ref{TTor}, we shall show explicit expressions for the first
order, when particular bases are chosen.

Consider two pairs of positive definite Hermitian forms $G_{1AB}$, $%
G_{2\alpha \beta }$ and $\widehat{G}_{1CD}$, $\widehat{G}_{2\alpha \beta }$
in the spaces $E$ and $X$. Since they are positive define, they are
equivalent, i.e. 
\begin{eqnarray}
\bar{U}_{~A}^{C}\widehat{G}_{1CD}U_{~B}^{D} &=&G_{1AB}  \label{1} \\
\bar{V}_{~\alpha }^{\gamma }\widehat{G}_{2\gamma \eta }V_{~\beta }^{\eta }
&=&G_{2\alpha \beta }  \label{2}
\end{eqnarray}

with some $U:E\rightarrow E$ and $V:X\rightarrow X$ invertible.

Consider now the linear operator $T_{~\alpha }^{A}:X\rightarrow E$ with $%
\dim \left( X\right) \leq \dim \left( E\right) .$ We call $\hat{\sigma}_{i}%
\left[ T_{~\alpha }^{A}\right] $ the singular values of $T_{~\alpha }^{A}$
defined using $\widehat{G}_{1CD}$, $\widehat{G}_{2\alpha \beta }$ and $%
\sigma _{i}\left[ T_{~\alpha }^{A}\right] $ using $G_{1AB}$, $G_{2\alpha
\beta }$.

\begin{lemma}
\label{Lemma_metricas} The operator $T_{~\alpha }^{A}\left( \varepsilon
\right) =K_{~\alpha }^{A}+\varepsilon e^{i\theta }B_{~\alpha
}^{A}:X\rightarrow E,$ has singular values $\sigma _{i}[T_{~\alpha
}^{A}]=O\left( \varepsilon ^{l_{i}}\right) $ with $i=1,...,\dim \left(
X\right) $ for some $l_{i}$ if and only if $\hat{\sigma}_{i}[T_{~\alpha
}^{A}]=O\left( \varepsilon ^{l_{i}}\right) $
\end{lemma}

\begin{proof}
If we call $\lambda _{i}[T^{\ast }\circ T]$ to the eigenvalues of $T^{\ast
}\circ T$ then%
\begin{equation}
\hat{\sigma}_{i}[T]=\sqrt{\lambda _{i}[T^{\ast }\circ T]}=\sqrt{\lambda
_{i}[V^{-1}\circ T^{\ast }\circ T\circ V]}=\sigma _{i}[U^{-1}\circ T\circ V]
\label{3}
\end{equation}%
The last equality is easy to prove.

Considering $T_{~\alpha }^{A}$ $:X\rightarrow E$ with rank $r$, we recall
that from definition of SVD 
\begin{equation*}
\sigma _{1}[T_{~\alpha }^{A}]\geq \sigma _{2}[T_{~\alpha }^{A}]\geq ...\geq
\sigma _{r}[T_{~\alpha }^{A}]>0=\sigma _{r+1}[T_{~\alpha }^{A}]=...=\sigma
_{\dim \left( X\right) }[T_{~\alpha }^{A}]
\end{equation*}

In \cite{stewart1990matrix} it is proved that 
\begin{equation*}
\sigma _{e}\left[ U^{-1}\right] \sigma _{u}\left[ V\right] \sigma
_{i}[T]\leq \sigma _{i}[U^{-1}\circ T\circ V]\leq \sigma _{i}[T]\sigma _{1}%
\left[ U^{-1}\right] \sigma _{1}\left[ V\right] \text{ \ }\forall i=1,...,u
\end{equation*}%
with $\dim X=u$, $\dim E=e$ and $\sigma _{e}\left[ U^{-1}\right] \sigma _{u}%
\left[ V\right] \neq 0$ \ since $U$ and $V$ are invertible. By this
expression, we see that if $\sigma _{i}[T]=O\left( \varepsilon
^{l_{i}}\right) $ then $\sigma _{i}[U^{-1}\circ T\circ V]=O\left(
\varepsilon ^{l_{i}}\right) $ $\forall i.$ Thus, due to equation (\ref{3}) $%
\hat{\sigma}_{i}[T]$ this is exactly $\sigma _{i}[U^{-1}\circ T\circ V]$ and
we conclude the proof.
\end{proof}

\begin{theorem}
\label{TTor}Let the operator be $T_{~\alpha }^{A}=K_{~\alpha
}^{A}+\varepsilon e^{i\theta }B_{~\alpha }^{A}:X\rightarrow E$ where $%
\varepsilon e^{i\theta }B_{~\alpha }^{A}$ \ represents a perturbation of $K$
with any $\theta \in \left[ 0,2\pi \right] $, $\varepsilon $ real, $0\leq
\left\vert \varepsilon \right\vert <<1$ and $K$ has rank $r$. Consider $%
T_{~\alpha }^{A}$ in basis in which $G_{1AB}=diag\left( 1,...,1\right) $ and 
$G_{2\alpha \beta }=diag\left( 1,...,1\right) $ and the SVD of $T$ is 
\begin{eqnarray*}
T_{~\beta }^{A} &=&\left( U_{2},U_{1},U_{3}\right) _{~i}^{A}\left( 
\begin{array}{cc}
\Sigma _{+} & 0 \\ 
0 & 0 \\ 
0 & 0%
\end{array}%
\right) _{~i^{\prime }}^{i}\delta _{2}^{i^{\prime }j^{\prime }}\left( \bar{V}%
_{2},\bar{V}_{1}\right) _{~j^{\prime }}^{\tau }\delta _{2\tau \beta } \\
&=&\left( U_{2},0,0\right) _{~i}^{A}\left( 
\begin{array}{cc}
\Sigma _{+} & 0 \\ 
0 & 0 \\ 
0 & 0%
\end{array}%
\right) _{~i^{\prime }}^{i}\delta _{2}^{i^{\prime }j^{\prime }}\left( \bar{V}%
_{2},0\right) _{~j^{\prime }}^{\tau }\delta _{2\tau \beta }
\end{eqnarray*}

Therefore

1)If $\sigma _{i}\left[ K_{~\beta }^{A}\right] >0$ \ \ $i=1,...,r$ \ are the
singular values of $K,$ then the singular value of $T_{~\alpha }^{A}$ can be
expanded as 
\begin{equation*}
\sigma _{i}\left[ T_{~\alpha }^{A}\right] =\sigma _{i}\left[ K_{~\beta }^{A}%
\right] +\left\vert \varepsilon \right\vert \xi _{i}+O\left( \varepsilon
^{2}\right) \text{ \ \ \ \ with }i=1,...,r
\end{equation*}

for some $\xi _{i}$ (see \cite{soderstrom1999perturbation} for explicit
formulas)

2) When $K_{~\beta }^{A}$ has $\sigma _{i}[K_{~\beta }^{A}]=0$ for $%
i=r+1,...,u$ \ then the corresponding $u-r$ singular values of $T_{~\alpha
}^{A}$ are 
\begin{eqnarray}
\sigma _{i}\left[ K_{~\alpha }^{A}+\varepsilon e^{i\theta }B_{~\alpha }^{A}%
\right] &=&0+\left\vert \varepsilon \right\vert ~\sigma _{i}\left[ \delta
_{1}^{lj}\left( 0,\bar{U}_{1},\bar{U}_{3}\right) _{~j}^{C}\delta
_{1CD}B_{~\alpha }^{D}\left( 0,V_{1}\right) _{~m}^{\alpha }\right] +O\left(
\varepsilon ^{2}\right)  \label{TTor_1} \\
\text{with }j &=&1,..,e;\text{ }l=r+1,...,e\text{, }m=r+1,...,u\text{ and }%
i=r+1,...,u  \notag
\end{eqnarray}

\footnote{\label{footnote_deriv} These singular values are not
differentiable respect to $\varepsilon $ in $\varepsilon =0,$ due to the
presence of the module $\left\vert \varepsilon \right\vert $. However, it is
possible differentiate respect to the module.} Notice that equation (\ref%
{TTor_1}) does not depend on $\theta $.
\end{theorem}


\bibliographystyle{unsrt}
\bibliography{paper2}

\end{document}